\documentclass{lmcs}
\pdfoutput=1

\usepackage{lastpage}
\lmcsdoi{19}{4}{34}
\lmcsheading{}{\pageref{LastPage}}{}{}%
{Nov.~15,~2022}{Dec.~20,~2023}{}

\usepackage[LGR,T1]{fontenc}
\usepackage[utf8]{inputenc}
\usepackage{alphabeta}
\usepackage[main=english,french]{babel}

\usepackage{amssymb}

\renewcommand{\leq}{\leqslant}
\renewcommand{\geq}{\geqslant}
\renewcommand{\emptyset}{\varnothing}

\usepackage{ebproof}
\newcommand{\regle}[1]{(#1)}
\ebproofset{
	right label template={$\regle{\inserttext}$}
	}

\usepackage{tikz}
	\usetikzlibrary{positioning}
	\usetikzlibrary{arrows.meta}

\usepackage{csquotes}
\usepackage{needspace}
\usepackage{subcaption}
\usepackage{xspace}

\AtEndPreamble{
\usepackage{xurl}
\hypersetup{breaklinks=true}
}

\AtEndPreamble{
\usepackage{cleveref}
\crefname{thm}{Theorem}{Theorems}
\crefname{lem}{Lemma}{Lemmas}
\crefname{cor}{Corollary}{Corollaries}
\crefname{defi}{Definition}{Definitions}
\crefname{rem}{Remark}{Remarks}
\crefname{exa}{Example}{Examples}
\crefname{conj}{Conjecture}{Conjectures}
\crefname{nota}{Notation}{Notations}
}

\usepackage[style=alphabetic]{biblatex}
\renewcommand{\bibfont}{\small}

\bibliography{simulation.bib}

\definecolor{R}{RGB}{244,7,110}
\definecolor{L}{RGB}{6,170,162}

\newcommand{\ie}{\emph{i.e.}\xspace}

\newcommand{\NN}{\mathbf{N}}

\newcommand{\MRel}{\mathbf{MRel}}

\newcommand{\Dd}{\mathcal{D}}

\newcommand{\Mm}{\mathcal{M}}
\newcommand{\Pp}{\mathcal{P}}
\newcommand{\Ss}{\mathcal{S}}
\newcommand{\Tt}{\mathcal{T}}
\newcommand{\Vv}{\mathcal{V}}

\newcommand{\sS}{\mathfrak{S}}

\newcommand{\relcomp}{\mathop{\text{\textbf{;}}}}

\newcommand{\flt}{\mathrel{\widetilde{\longrightarrow}}}

\newcommand{\defeq}{\coloneqq}

\newcommand{\ptnewline}{\\[4ex]}
\NewDocumentEnvironment{prooftabular}{o}{%
	\vskip 2ex
	\begin{center}
	\begin{tabular}{c@{\IfValueTF{#1}{#1}{\qquad}}c}
}{%
	\end{tabular}
	\end{center}
	\vskip 2ex
}

\newcommand{\ax}{\mathrm{ax}}

\newcommand{\coI}{\mathrm{coI}}

\newcommand{\later}{\mathop\triangleright}

\newcommand{\linf}[1]{\texorpdfstring%
	{\Lambda_{\,\infty}^{#1}}%
	{lambda-infinity-#1}
}
\newcommand{\bisim}{=}
\newcommand{\nbisim}{\neq}

\newcommand{\lbinf}[1]{\Lambda_{\,\infty\bot}^{#1}}
\newcommand{\BT}[1]{\mathrm{BT}\left(#1\right)}
\newcommand{\eqbohm}{\mathrel{=_{\mathcal{B}}}}

\newcommand{\FV}[1]{\mathrm{FV}\left(#1\right)}

\DeclareDocumentCommand{\subst}{mmo}{%
	#1 \!\left[ #2 / \IfValueTF{#3}{#3}{x} \right]
}

\newcommand{\flb}{\mathrel{\longrightarrow_\beta}}
\newcommand{\flbs}{\mathrel{\longrightarrow_\beta^*}}
\newcommand{\flbi}{\mathrel{\longrightarrow_\beta^\infty}}

\newcommand{\flbg}[1]{\mathrel{\longrightarrow_{\beta\geq #1}}}
\newcommand{\flbgs}[1]{\mathrel{\longrightarrow_{\beta\geq #1}^*}}
\newcommand{\flbe}[1]{\mathrel{\longrightarrow_{\beta=#1}}}
\newcommand{\flbes}[1]{\mathrel{\longrightarrow_{\beta=#1}^*}}
\newcommand{\flbig}[1]{\mathrel{\longrightarrow_{\beta\geq #1}^\infty}}
\newcommand{\flbb}{\mathrel{\longrightarrow_{\beta\bot}}}
\newcommand{\flbbr}{\mathrel{\longrightarrow_{\beta\bot}^?}}
\newcommand{\flbbs}{\mathrel{\longrightarrow_{\beta\bot}^*}}
\newcommand{\flbbi}{\mathrel{\longrightarrow_{\beta\bot}^\infty}}
\newcommand{\flh}{\mathrel{\longrightarrow_h}}
\newcommand{\flhs}{\mathrel{\longrightarrow_h^*}}

\usepackage[only,llparenthesis,rrparenthesis,llbracket,rrbracket]{stmaryrd}
\newcommand{\context}[2]{#1\llparenthesis #2\rrparenthesis}
\newcommand{\hole}{*}

\newcommand{\lr}[1]{\Lambda_r^{#1}}
\newcommand{\rapp}[1]{\left\langle#1\right\rangle}
\newcommand{\Mfin}{\Mm_{\mathrm{fin}}}
\newcommand{\ms}[1]{\bar{#1}}
\newcommand{\mscard}{\text{\texttt{\#}}}

\newcommand{\freemod}[2]{#1\langle #2 \rangle}
\DeclareDocumentCommand{\lrsums}{o}{%
	\freemod{2}{\lr{\IfValueT{#1}{#1}}}%
}

\DeclareDocumentCommand{\rsubst}{mmo}{%
	#1 \!\left\langle #2 / \IfValueTF{#3}{#3}{x} \right\rangle
}

\newcommand{\flro}{\mathrel{\longmapsto_r}}
\newcommand{\flror}{\mathrel{\longmapsto_r^?}}
\newcommand{\flr}{\mathrel{\longrightarrow_r}}

\newcommand{\flrs}{\mathrel{\longrightarrow_r^*}}

\newcommand{\flrgo}[1]{\mathrel{\longmapsto_{r\geq #1}}}
\newcommand{\flrg}[1]{\mathrel{\longrightarrow_{r\geq #1}}}

\newcommand{\flrgs}[1]{\mathrel{\longrightarrow_{r\geq #1}^*}}
\newcommand{\flrt}{\mathrel{\flt_r}}

\newcommand{\flrst}{\mathrel{\widetilde{\flrs}}}

\newcommand{\flrgst}[1]{\mathrel{\widetilde{\flrgs{#1}}}}
\newcommand{\flrho}{\mathrel{\longmapsto_{rh}}}
\newcommand{\flrh}{\mathrel{\longrightarrow_{rh}}}

\newcommand{\nf}[2]{\mathrm{nf}_{#1}\!\left(#2\right)}
\newcommand{\nft}[2]{\widetilde{\mathrm{nf}}_{#1}\!\left(#2\right)}

\newcommand{\rcontext}[2]{#1\llbracket #2\rrbracket}

\newcommand{\lrpl}[2]{\Lambda_{r}^{#2+#1}}

\newcommand{\apptay}{\mathrel{\ltimes}}
\newcommand{\Tay}{\Tt}
\newcommand{\eqtay}{\mathrel{=_{\Tay}}}
\newcommand{\apptaymax}[1]{\apptay_{\hspace*{-1pt}< #1}}
\newcommand{\Taymax}[1]{{\Tt\hspace{-3pt}}_{< #1}}
\newcommand{\Taymaxbang}[1]{\Tt^!_{< #1}}

\newcommand{\myorjourn}[2]{#2}	

\begin{document} 

\title[%
	Finitary simulation of infinitary 
	\texorpdfstring{\MakeLowercase{β}}{Beta}-reduction
	via Taylor expansion
]{%
	Finitary simulation of infinitary 
	\texorpdfstring{\MakeLowercase{β}}{Beta}-reduction
	\texorpdfstring{\\}{}
	via Taylor expansion, and applications
}

\ifcsname bibfont\endcsname%
  \renewcommand*{\bibfont}{\normalfont\bibliofont}
\else
\fi

\author[R.~Cerda]{Rémy Cerda\lmcsorcid{0000-0003-0731-6211}}[a]
\author[L.~Vaux~Auclair]{Lionel Vaux Auclair\lmcsorcid{0000-0001-9466-418X}}[a]
\address{Aix-Marseille Université, CNRS, I2M}
\email{%
	\href{mailto:remy.cerda@univ-amu.fr}{\tt remy.cerda@univ-amu.fr},
	\href{mailto:lionel.vaux@univ-amu.fr}{\tt lionel.vaux@univ-amu.fr}
}

\keywords{lambda-calculus, infinitary rewriting, taylor expansion, program approximation, semantics of program languages}

\begin{abstract}
Originating in Girard's Linear logic, Ehrhard and Regnier's Taylor expansion of λ-terms has been broadly used as a tool to approximate the terms of several variants of the λ-calculus. Many results arise from a Commutation theorem relating the normal form of the Taylor expansion of a term to its Böhm tree. This led us to consider extending this formalism to the infinitary λ-calculus, since the $\linf{001}$ version of this calculus has Böhm trees as normal forms and seems to be the ideal framework to reformulate the Commutation theorem.

We give a (co-)inductive presentation of $\linf{001}$. We define a Taylor expansion on this calculus, and state that the infinitary β-reduction can be simulated through this Taylor expansion. The target language is the usual resource calculus, and in particular the resource reduction remains finite, confluent and terminating. Finally, we state the generalised Commutation theorem and use our results to provide simple proofs of some normalisation and confluence properties in the infinitary λ-calculus.
\end{abstract}

\maketitle

\tableofcontents


\section{Introduction}

The seminal idea of \emph{quantitative semantics}, introduced in the early 1980s by Girard as an alternative to traditional denotational semantics based on Scott domains, is to interpret the terms of the λ-calculus by power series \autocites{Girard88}[for a bief survey see][]{Pagani14}. In this model, each monomial of the interpretation captures a finite approximation of the execution of the interpreted term, and its degree corresponds to the number of times it uses its argument. The parallelism between the decomposition of such a power series into linear maps and the behaviour of the cut-elimination of proofs led Girard to introduce \emph{linear logic} \autocite{Girard87, Girard95}, which has been a major and fruitful refinement of the Curry-Howard correspondence.

In the early 2000s, Ehrhard reformulated Girard's quantitative semantics in a more standard algebraic framework, where terms are interpreted as analytic maps between certain vector spaces \autocite{Ehrhard05}. The notion of differentiation, that is available in this framework, was then brought back to the syntax by Ehrhard and Regnier in their \emph{differential λ-calculus} \autocite{EhrhardRegnier03}. Eventually, they defined the operation of \textit{Taylor expansion} which maps λ-terms to infinite sums of resource terms — the latter are the terms of the resource λ-calculus, which is the finitary, purely linear fragment of the differential λ-calculus. Each term of the sum thus gives a finite approximation of the operational behaviour of the original term \autocites{EhrhardRegnier08}[for a lightened presentation]{BarbarossaManzonetto20}.

The strength of this tool is the strong normalisation property of resource terms, and the fact that Taylor expansion commutes with normalisation:
the normal form of the Taylor expansion of a term is the Taylor expansion of the Böhm tree of this term \autocite{EhrhardRegnier06}.
This Commutation theorem enables one to deduce properties of some λ-terms from 
the properties of their Taylor expansion;
typically, properties of the (possibly) non-terminating execution of a λ-term,
previously characterized by coinductive objects like its Böhm tree,
are proved \emph{via} the Taylor expansion by mere induction.
This approach has been successfully applied not only to the ordinary λ-calculus \autocites{EhrhardRegnier06,Olimpieri18,Olimpieri20,BarbarossaManzonetto20}, but also to nondeterministic \autocite{BucciarelliAl12,Vaux19}, probabilistic \autocite{DalLagoZorzi12,DalLagoLeventis19}, call-by-value \autocite{KerinecAl20}, and call-by-push-value \autocite{ChouquetTasson20} calculi, as well as Parigot's λμ-calculus \autocite{Barbarossa22}.

\medskip

In this paper, we aim to extend this formalism
to the \emph{infinitary} λ-calculus. 
Böhm trees \autocite{Bohm68, Barendregt77} 
were already a kind of infinitary λ-terms, 
but an infinitary calculus (having infinite terms and infinite reductions) 
was first introduced in the 1990s by Kennaway, Klop, Sleep 
and de~Vries \autocite{KennawayAl95, KennawayAl97} 
and by Berarducci \autocite{Berarducci96}. 
Initially presented as the metric completion 
of the set of λ-terms (considered as finite syntactic trees), 
the set of infinite λ-terms has been reformulated
as an ideal completion \autocite{Bahr18}, 
and maybe more crucially as the \enquote{coinductive version} 
of the λ-calculus \autocite{Joachimski04,EndrullisAl13, Czajka14}.

Even if the \enquote{plain} infinitary λ-calculus does not enjoy confluence, several results of confluence and of normalisation modulo \enquote{meaningless} terms have been established \autocite{KennawayAl97, Czajka14, Czajka20}, as well as a standardisation theorem using coinductive techniques \autocite{EndrullisAl13}. Some normalisation properties have also been characterised using non-idem\-potent intersection types \autocite{Vial17, Vial21}.

These results are often only established in one of the different variants of the infinitary λ-calculus. Indeed, Kennaway \emph{et al.} identify eight variants depending on the metric one chooses on syntactic trees (each of the three constructors of λ-terms can \enquote{add depth} to the term), among which only three enjoy reasonable properties (in addition to the finitary variant). Following the authors, they are called $\linf{001}$, $\linf{101}$ and $\linf{111}$. In the following, we will concentrate on the $\linf{001}$ variant, that is the one where a term can have an infinite branch only if its right applicative depth tends to infinity. 

The motivation for choosing $\linf{001}$ is that the normal forms of this calculus are the Böhm trees which are, as we recalled above, strongly related to the Taylor expansion.
In some sense, $\linf{001}$ is the \enquote{natural} setting to define and manipulate the Taylor expansion of ordinary λ-terms, as we hope to advocate for. Indeed, it enables us to state Ehrhard and Regnier's Commutation theorem without any particular definition of the Taylor expansion of Böhm trees, and then to prove some classical results that are preserved in $\linf{001}$, like characterisations of head- and β-normalisation or the Genericity lemma.

\medskip

Since we want to take advantage of the modern, coinductive approach of \autocite{EndrullisAl13}, we will provide a definition of $\linf{001}$ using coinduction. However, some technicalities arise from the fact that this is not the \enquote{fully coinductive version} of λ-calculus and that one has to mix induction and coinduction to manipulate terms and reductions in $\linf{001}$.

Such mixings are not new and have been appearing in various areas for several decades. In particular, type systems featuring inductive and coinductive types have been presented in the late 1980s by Hagino and Mendler \autocite{Hagino87,Mendler91}, and even Eratosthenes' sieve can be seen as an inductive-coinductive structure \autocite{Bertot05}. A wide range of examples is provided by Basold's PhD thesis, which builds a whole type-theoretic framework for inductive-coinductive reasoning \autocite{Basold18}. Several previous formalisations of mixed induction and coinduction had been proposed, in particular in \autocites{DanielssonAltenkirch09}{Czajka19}{DalLago16}. We will provide a mixed formal system inspired by the latter.

\paragraph{Contributions and structure of the paper.}

In \cref{linf}, we recall the definition of the infinitary λ-calculus
$\linf{001}$ and discuss the setting we choose for this mixed
inductive-coinductive construction.

In \cref{taylor}, we extend the Taylor expansion of λ-terms to this calculus,
and we show our main result in \cref{simul}: the reduction of the Taylor
expansion provides a simulation of the infinitary β-reduction
(\cref{simul:thm:simul_flbi}) by a (possibly infinite) superposition of finite
reductions.
This adapts the results of the second author for the ordinary β-reduction
\autocite{Vaux17,Vaux19}, with three important differences, induced by the infinitary
nature of our setting:
\begin{enumerate}
	\item In \autocite{Vaux19}, one step of β-reduction is simulated by a
		superposition of single steps of \emph{parallel} reduction on
		resource terms (reducing all the copies of the fired redex
		simultaneously in each component).
		Here, since one step of infinitary β-reduction amounts to a possibly
		infinite sequence of single β-reductions, we can no longer bound the
		number of reductions to be performed on resource terms in the
		simulation.
		This forces us to consider the reflexive and transitive closure of
		reduction, rather than parallel reduction, as the underlying dynamics
		on resource terms.
	\item Due to the previous point, establishing the simulation is more
		demanding.
		We first obtain the simulation of finite sequences of reductions
		exactly as in \autocite{Vaux19}.
		Then we exploit the fact that the depth of fired redexes in an infinite
		β-reduction must tend to infinity, together with the strictly finitary
		nature of the reduction of resource terms (in particular, the fact 
		that the size of resource terms is nonincreasing under reduction):
		we deduce the general simulation result by a kind of diagonal
		argument (see \cref{simul:sub:diag}).
	\item Also due to the first point, we can no longer consider arbitrary
		infinite weighted sums of resource terms as the target of Taylor
		expansion: reducing arbitrarily the components of such a sum
		might yield infinite sums of coefficients (an example
		is given in \cref{taylor:rem:quantitative}).
		We thus choose to keep to a \emph{qualitative} setting only,
		\emph{i.e.} use boolean coefficients:
		equivalently, we consider sets of resource terms, rather than sums.
\end{enumerate}

Finally, in \cref{norm}, we use this framework to prove the
Commutation theorem (\cref{norm:thm:commut}).
We show that this provides new proofs of the classical results of normalisation
(\cref{norm:lem:bbot_wn}) and confluence (\cref{norm:cor:bbot_confl}),
as well as characterisations of solvability and normalisation in $\linf{001}$
similar to the ones known for the ordinary λ-calculus.
We also show an infinitary Genericity lemma (\cref{norm:thm:genericity}),
adapting the technique introduced by \autocites{BarbarossaManzonetto20}.

\section{The infinitary \texorpdfstring{{\upshape λ}}{lambda}-calculus}
\label{linf}

\subsection{The set $\linf{001}$ of 001-infinitary \texorpdfstring{λ}{lambda}-terms}


The original definition of the infinitary λ-calulus by Kennaway \emph{et al.} \autocite*{KennawayAl97} was topological. Finite terms were represented by their syntactic tree, and the usual distance $d$ on trees was defined on them by:
	\[ d(M,N) = 2^{-(\text{the smallest depth at which $M$ and $N$ differ})}. \]
The space of infinitary λ-terms was obtained by taking the metric completion. 

One can notice that this definition is dependent on the notion of \emph{depth}. Indeed, the authors defined eight variants of $\linf{}$, each one of them using a different notion of depth of (an occurrence of) a subterm $N$ in a term $M$:
	\begin{flalign*}
	&& \mathrm{depth}^{abc}_N(N) &= 0, \\
	&& \mathrm{depth}^{abc}_{\lambda x.M}(N) &= a+\mathrm{depth}^{abc}_M(N), \\
    && \mathrm{depth}^{abc}_{(M)M'}(N) &= b+\mathrm{depth}^{abc}_M(N)
    	& \text{if the occurrence is in $M$,}\\
    && \mathrm{depth}^{abc}_{(M)M'}(N) &= c+\mathrm{depth}^{abc}_{M'}(N)
    	& \text{otherwise,}
	\end{flalign*}
where $a,b,c\in\{0,1\}$. 
This gives rise to eight spaces $\linf{abc}$, 
where $\linf{000}$ is the set of finite λ-terms $\Lambda$ 
and $\linf{111}$ contains all infinitary λ-terms 
(notice that \emph{infinitary} terms are not necessarily \emph{infinite}, 
since all finite terms also belong to the spaces defined). 
The depth of all infinite branches in a term of $\linf{abc}$ 
must go to infinity, that is to say 
such a branch must cross infinitely often a node increasing the depth. 
In particular, for the $\linf{001}$ version we are interested in, 
the only infinite branches allowed 
are those crossing infinitely often the right side of an application. 
In \cref{linf:fig:001oupas}, the left term is in $\linf{001}$ 
whereas the right one is not
(notice that in the former term, the infinite branch also
crosses infinitely many lambdas;
this is \emph{not} forbidden, provided this infinite branch
crosses infinitely many right sides of applications).

\begin{figure}\begin{center}
	\raisebox{-0.5\height}{
	\begin{tikzpicture}
	\path	( 0, 0)	node (lf)	{ $\lambda f$ }
	++		( 0,-1)	node (l0)	{ $\lambda x_0$ }
	++		( 0,-1)	node (a0)	{ $@$ }
	+		(-1,-1)	node (f0)	{ $f$ }
	++		( 1,-1)	node (l1)	{ $\lambda x_1$ }
	++		( 0,-1)	node (a1)	{ $@$ }
	+		(-1,-1)	node (f1)	{ $f$ }
	++		( 1,-1)	node (l2)	{ $\lambda x_2$ }
	++		( 0,-1)	node (a2)	{ $@$ }
	+		(-1,-1)	node (f2)	{ $f$ }
	+		( 1,-1)	node (inf)	{};
	\draw (lf) -- (l0);
	\draw (l0) -- (a0);
	\draw (a0) -- (f0);
	\draw (a0) -- (l1);
	\draw (l1) -- (a1);
	\draw (a1) -- (f1);
	\draw (a1) -- (l2);
	\draw (l2) -- (a2);
	\draw (a2) -- (f2);
	\draw [dotted] (a2) -- (inf);
	\end{tikzpicture}
	}
	\quad
	\raisebox{-0.5\height}{
	\begin{tikzpicture}
	\path 	( 0, 0)	node (l0)	{ $\lambda x_0$ }
	++		( 0,-1)	node (a0)	{ $@$ }
	+		( 1,-1)	node (x0)	{ $x_0$ }
	++		(-1,-1)	node (l1)	{ $\lambda x_1$ }
	++		( 0,-1)	node (a1)	{ $@$ }
	+		( 1,-1)	node (x1)	{ $x_1$ }
	++		(-1,-1)	node (l2)	{ $\lambda x_2$ }
	++		( 0,-1)	node (a2)	{ $@$ }
	+		( 1,-1)	node (x2)	{ $x_2$ }
	++		(-1,-1)	node (inf)	{};
	\draw (l0) -- (a0);
	\draw (a0) -- (x0);
	\draw (a0) -- (l1);
	\draw (l1) -- (a1);
	\draw (a1) -- (x1);
	\draw (a1) -- (l2);
	\draw (l2) -- (a2);
	\draw (a2) -- (x2);
	\draw [dotted] (a2) -- (inf);
	\end{tikzpicture}
	}
	\caption{Two infinitary terms, only the left one of which is 
	001-infinitary.}
	\label{linf:fig:001oupas}
\end{center}\end{figure}

All versions enjoy weak normalisation, 
provided one identifies all \enquote{0-active} terms 
(\emph{i.e.} those terms such that every reduct contains a redex at depth~0) 
with a single constant~$\bot$.
For instance, in $\linf{000}$, this means identifying all non normalising terms;
and in $\linf{001}$ this means identifying all non head-normalising terms.
With this extended reduction, only three versions enjoy confluence, and thus unicity of normal forms: $\linf{001}$, $\linf{101}$ and $\linf{111}$.
Their respective normal forms are three already known notions of 
infinite expansions of a term, namely 
Böhm trees \autocites{Barendregt77}[§~2.1.13]{Barendregt84}, 
Lévy-Longo trees \autocite{Levy75,Longo83,Ong88} 
and Berarducci trees \autocite{Berarducci96}.
The two latter equate less terms than the unsolvable ones, and thus provide a more fine-grained description of the computational behaviour of λ-terms. 


As an alternative, the infinitary λ-calculus can be seen as 
the \enquote{coinductive version} of the λ-calculus. 
If the set $\Lambda$ of the λ-terms is built inductively on the signature:
	\begin{equation*}
		M, N, \ldots \ \defeq \ 
		x\in\Vv \ |\ \lambda x. M \ |\ (M)N,
	\tag{$\sigma$}
	\end{equation*}
given a fixed set $\Vv$ of variables 
(that is, it is the initial algebra of the 
corresponding monotonous functor $\Vv \ +\ \lambda\Vv.- \ +\ (-)- : 
\mathbf{Set}^3 \to \mathbf{Set}$), 
then the set $\linf{}$ of all infinitary 
λ-terms is built coinductively on the same signature, 
as the terminal coalgebra 
of the same functor \autocite[for a detailed reminder of these constructions, 
see for instance][]{AdamekAl18}. This construction is summarised in the 
following notation, using fix-points:
	\[
	\Lambda = \mu X.( \Vv + \lambda\Vv.X + (X)X ) \qquad
	\linf{} = \nu X.( \Vv + \lambda\Vv.X + (X)X ).
	\]

This coinductive approach has been fruitfully exploited by Endrullis and Polonsky \autocite{EndrullisAl13} and Czajka \autocite{Czajka14,Czajka20} in the case of $\linf{111}$. We would like to use it in the case of $\linf{001}$, but this implies mixing induction and coinduction in order to distinguish between the \enquote{allowed} and \enquote{forbidden} infinite branches. Thus, using the same notation as above, we provide the following definition.


\begin{defi}[001-infinitary terms] \label{linf:def:l001}
	Given a fixed set of variables $\Vv$, the set $\linf{001}$ of 
	\emph{001-infinitary λ-terms} is defined by:
	\[ \linf{001} = \nu Y. \mu X.( \Vv + \lambda\Vv.X + (X)Y ). \]
\end{defi}

It is beyond the scope of this paper to describe 
a general framework for defining and manipulating 
such a mixed inductive-coinductive set. 
One may consider the type-theoretic system built 
by Basold in his extensive study of this question \autocite{Basold18}. 
As a somehow less technological alternative, 
we interpret the binders $\mu$ and $\nu$ as 
the usual least and greatest fix-point constructions 
in the lattice $(\Pp(\linf{}), \subseteq)$, 
or the initial algebra and terminal coalgebra of the
functors $\Vv \ +\ \lambda\Vv.- \ +\ (-)-$ 
and $\mu X.( \Vv + \lambda\Vv.X + (X)- )$ respectively.


\cref{linf:def:l001} can be unfolded using a mixed formal system (in such a system, simple bars denote inductive rules and double bars denote coinductive rules). This reformulation, inspired by \autocite{DalLago16}, provides a graphical description of terms in $\linf{001}$.

\begin{defi}[001-infinitary terms, using a mixed formal system] \label{linf:def:formsys}
	$\linf{001}$ is the set of all coinductive terms $T$ on the signature 
	$\sigma$\footnote{%
		Notice that for any $M$,
		$\later M$ cannot be a term in $\linf{001}$,
		the only valid derivations being those producing
		terms on the signature $\sigma$,
		\ie ending with one of the first three rules.
	}
	such that $\ \vdash T$ can be derived in the following system:

	\begin{center}
		\begin{prooftree}[center=false]
			\infer0[\Vv]{ \vdash x }
		\end{prooftree}
		\qquad
		\begin{prooftree}[center=false]
			\hypo{ \vdash M }
			\infer1[\lambda]{ \vdash \lambda x.M }
		\end{prooftree}
		\qquad
		\begin{prooftree}[center=false]
			\hypo{ \vdash M }
			\hypo{ \vdash \later N }
			\infer2[@]{ \vdash (M)N }
		\end{prooftree}
		\qquad
		\begin{prooftree}[center=false]
			\hypo{ \vdash M }
			\infer[double]1[\coI]{ \vdash \later M }
		\end{prooftree}
	\end{center}
\end{defi}

\begin{rem}
	To make the coinductive step explicit, we use the \emph{later} modality 
	$\later$ due to \autocite{Nakano00} and named after \autocite{AppelAl07}.
	This formalism could be condensed in the following \enquote{mixed rule} 
	$\regle{@'}$, in a rather unusual fashion:
	\begin{center}
	\begin{prooftree}
		\hypo{ \vdash M }
		\hypo{ \vdash N }
		\infer1{\hspace*{1.2cm}}
		\rewrite{
			\hspace*{\myorjourn{-2.1mm}{-2.5mm}}
			\raisebox{\myorjourn{-1.7mm}{-1.5mm}}{\box\treebox} }
		\infer2[@']{ \vdash (M)N }
	\end{prooftree}
	\end{center}
\end{rem}

\begin{exa}
	Have $Y^* \defeq \lambda f.(f)(f)(f)\dots$, 
	which can be defined coinductively as $Y^* \defeq \lambda f.f^\infty$ 
	where $f^\infty$ is the largest solution 
	of the equation $f^\infty = (f)f^\infty$. 
	This corresponds to the derivation in \cref{linf:fig:exemples:ystar}.
	Similarly, one can show that $(f^\infty) f^\infty \in \linf{001}$
	and $(f^\infty)^\infty \in \linf{001}$,
	as derived in \cref{linf:fig:exemples:finf2,linf:fig:exemples:finfinf}.
\end{exa}

\begin{figure} \centering
	\begin{subfigure}[b]{.28\linewidth} \centering
		\begin{prooftree}
			\infer0{ \vdash f }
			\hypo{\parbox[t][0pt]{0pt}{
				\begin{tikzpicture}[baseline]
				\draw [-{To[scale=1.3]},dashed] (0,0)
				to[out=90,in=180] (0.8,0.5)
				to[out=0,in=90] (1.5,-1.2)
				to[out=270,in=0] (0.9,-1.7);
				\end{tikzpicture}
			}}
			\infer1{ \vdash f^\infty }
			\infer[double]1{ \vdash \later f^\infty }
			\infer2{ \vdash f^\infty = (f)f^\infty }
		\end{prooftree} \hspace*{0.5cm}
		\caption{} \label{linf:fig:exemples:finf}
	\end{subfigure}
	\begin{subfigure}[b]{.2\linewidth} \centering
		\begin{prooftree}
			\hypo{ \text{\subref{linf:fig:exemples:finf}} }
			\infer1{ \vdash f^\infty }
			\infer1{ \vdash Y^* = \lambda f.f^\infty }
		\end{prooftree}
		\caption{} \label{linf:fig:exemples:ystar}
	\end{subfigure}
	\begin{subfigure}[b]{.2\linewidth} \centering
		\begin{prooftree}
			\hypo{ \text{\subref{linf:fig:exemples:finf}} }
			\infer1{ \vdash f^\infty }
			\hypo{ \text{\subref{linf:fig:exemples:finf}} }
			\infer1{ \vdash f^\infty }
			\infer[double]1{ \vdash \later f^\infty }
			\infer2{ \vdash (f^\infty)f^\infty }
		\end{prooftree}
		\caption{} \label{linf:fig:exemples:finf2}
	\end{subfigure}
	\begin{subfigure}[b]{.28\linewidth} \centering
		\begin{prooftree}
			\hypo{ \text{\subref{linf:fig:exemples:finf}} }
			\infer1{ \vdash f^\infty }
			\hypo{\parbox[t][0pt]{0pt}{
				\begin{tikzpicture}[baseline]
				\draw [-{To[scale=1.3]},dashed] (0,0)
				to[out=90,in=180] (0.8,0.5)
				to[out=0,in=90] (1.5,-1.2)
				to[out=270,in=0] (0.9,-1.7);
				\end{tikzpicture}
			}}
			\infer1{ \vdash (f^\infty)^\infty }
			\infer[double]1{ \vdash \later (f^\infty)^\infty }
			\infer2{ \vdash (f^\infty)^\infty }
		\end{prooftree} \hspace*{0.5cm}
		\caption{} \label{linf:fig:exemples:finfinf}
	\end{subfigure}
	\caption{Some derivations corresponding to terms in $\linf{001}$.
	Notice that the loops are correct because they cross a coinductive rule.}
	\label{linf:fig:exemples}
\end{figure}

\begin{nota} \label{linf:nota:puissances}
	Given $M, N\in\linf{001}$ two terms and $k\in\NN$ an integer, we define:
	\[
	(M)N^{(k)} \defeq (\dots((M) \underbrace{N)\dots)N}_{\text{$k$ terms}}
	\qquad
	N^k \defeq \underbrace{(N)\dots (N)}_{\text{$k-1$ terms}} N
	\]
	The corresponding trees are described in \cref{linf:fig:puissances}. Notice that the term $N^\infty$ introduced in the previous example is coherent with this notation, whereas there is no possible $(M)N^{(\infty)}$ in $\linf{001}$.
\end{nota}

\begin{figure}\begin{center}
	\begin{tikzpicture}[node distance=15mm]
	\node (ak) { $@$ };
	\node (a2) [below left of=ak] { $@$ };
	\draw [dotted] (ak) -- (a2);
	\node (a1) [below left of=a2] { $@$ };
	\draw (a2) -- (a1);
	\node (m) [below left of=a1] { $M$ };
	\draw (a1) -- (m);
	\node (n1) [below right of=a1] { $N$ };
	\draw (a1) -- (n1);
	\node (n2) [below right of=a2] { $N$ };
	\draw (a2) -- (n2);
	\node (nk) [below right of=ak] { $N$ };
	\draw (ak) -- (nk);
	\end{tikzpicture}
	\hspace{2cm}
	\begin{tikzpicture}[node distance=15mm]
	\node (a1) { $@$ };
	\node (a2) [below right of=a1] { $@$ };
	\draw (a1) -- (a2);
	\node (ak) [below right of=a2] { $@$ };
	\draw [dotted] (a2) -- (ak);
	\node (nk) [below right of=ak] { $N$ };
	\draw (ak) -- (nk);
	\node (nkm1) [below left of=ak] { $N$ };
	\draw (ak) -- (nkm1);
	\node (n2) [below left of=a2] { $N$ };
	\draw (a2) -- (n2);
	\node (n1) [below left of=a1] { $N$ };
	\draw (a1) -- (n1);
	\end{tikzpicture}
	\caption{The terms $(M)N^{(k)}$ and $N^k$.}
	\label{linf:fig:puissances}
\end{center}\end{figure}


\subsection{What about \texorpdfstring{α}{alpha}-equivalence?}

As one usually does when working with λ-terms, we consider the terms up to α-equiva\-lence (renaming of bound variables) in the following. In particular, we will define substitution using Barendregt's variable convention, that is considering that any term has disjoint bound and free variables, which is usually achieved by renaming conflictual bound variables with fresh ones \autocite[§~2.1.13]{Barendregt84}. 

However, this requires some precautions in an infinitary setting 
since we could consider an infinite term $M$ such that $\FV{M}=\Vv$,
which would prevent us from taking a fresh variable. 
This obstacle can be overcome using some tricks, 
like taking a non-countable variable set $\mathcal{V}$, 
or ordering it to be able to implement Hilbert's hotel 
— which is usually done when the proofs are formalised using De~Bruijn indices 
\autocite{deBruijn72,EndrullisAl13, Czajka20}.

One can also use nominal sets \autocite{GabbayPitts02, Pitts13} to directly define the quotient of the infinitary λ-calculus modulo α-equivalence as the terminal coalgebra for some functor. This construction yields a corecursion principle allowing to define substitution and normal forms for $\linf{111}$ \autocite{KurzPetrisanAl12}. There is hope that the same tools could be applied to the specific case of $\linf{001}$.

One more solution, which seems radical but which we believe is appropriate in practice, is to restrict ourselves to infinitary terms whose subterms all contain a finite number of free variables. This makes it easy to get fresh variables to implement Barendregt's convention, while preserving the strength of $\linf{001}$ as a tool to study the infinite behaviour of finite λ-terms. Indeed, all infinitary terms generated by reductions of finitary ones enjoy this property of having finitely many free variables \autocite[Thm.~10.1.23]{Barendregt84}.%
\footnote{Note that, writing $\Lambda(\Gamma)$ for the set of λ-terms whose free variables are in the set $\Gamma$, it is clear that $\Lambda$ is the union of the sets $\Lambda(\Gamma)$ where $\Gamma$ ranges over finite sets of variables.
By contrast, the union of the sets $\linf{abc}(\Gamma)$ --- again, for $\Gamma$ finite --- is a strict subset of $\linf{abc}$, as introduced before.}

For the sake of simplicity, we stick to the presentation using a single class of variables (instead of relying on De Bruijn indices or nominal techniques), and assume without justification that it is always possible to obtain fresh variables.
We believe this question is completely orthogonal to the main matter of the paper anyway, and that our developments could be adapted straightforwardly to any other formalisation of bound variables.

\subsection{Finitary \texorpdfstring{β}{beta}-reduction}

The finitary β-reduction is defined exactly as in the usual λ-calculus. We just have to check that our definitions are consistent with the restrictions we put on infinitary terms.


\begin{defi}[substitution]
	Given $N\in \linf{001}$ and $x\in\Vv$, the \emph{substitution} $\subst 
	{\mathord{-}} N$ of $x$ by $N$ is the operation on terms defined as follows:
	\begin{flalign*}
		&& \subst x N					&\defeq N\\
		&& \subst y N					&\defeq y
			&\text{if $y\neq x$}\\
		&& \subst {(\lambda y.M)} N	&\defeq \lambda y.\subst M N
			&\text{by choosing $y \notin \FV{N}$}\\
		&& \subst {((M)M')} N			&\defeq (\subst M N) \subst {M'} N
	\end{flalign*}
\end{defi}
Note that this definition is not merely by induction, since we consider infinitary terms.
To be formal, given a derivation of ${}\vdash M$, we define a derivation of 
some judgement ${}\vdash M'$, and then set $\subst M N \defeq M'$.
To do so, we build the derivation of ${}\vdash \subst M N $ coinductively, 
following the derivation of ${}\vdash M$;
and inside each coinductive step, we proceed by induction on the finite tree of 
rules other than $\regle{\coI}$ at the root of the derivation of ${}\vdash M$:
	
\begin{itemize}
	\item Case $\regle{\Vv}$.
		Either $M\bisim x$,
		in which case we set $\subst M N \defeq N$
		and derive ${}\vdash \subst M N $ just like ${}\vdash N$;
		or $M\bisim y$ for some $y \neq x$
		and we set $\subst M N \defeq y$
		and derive ${}\vdash \subst M N $ by $\regle{\Vv}$.
	
	\item Case $\regle{\lambda}$.
		We have $M\bisim \lambda y.M'$, where ${}\vdash M'$ and we choose $y 
		\notin \FV{N}$.
		The induction hypothesis applies to the derivation of ${}\vdash M'$, 
		which gives ${}\vdash \subst {M'} N$,
		and we derive ${}\vdash \subst M N $ by $\regle{\lambda}$,
		setting $\subst M N \defeq \lambda y.\subst {M'} N$.
	
	\item Case $\regle{@}$. We have $M\bisim (M')M''$ and the derivation:
		\begin{center}\begin{prooftree}
		\hypo{ \vdots }
		\infer1{\vdash M' }
		\hypo{ \vdots }
		\infer1{\vdash M'' }
		\infer[double]1[\coI]{\vdash \later M'' }
		\infer2[@]{\vdash M }
		\end{prooftree}\end{center}

		As in the previous case,
		the induction hypothesis applies to the derivation of ${}\vdash M'$, 
		which gives ${}\vdash \subst {M'} N$.
		Moreover, under the guard of rule $\regle{\coI}$,
		we apply the construction coinductively,
		which yields a derivation of ${}\vdash\later \subst {M''} N$
		from the derivation of ${}\vdash\later M''$.
		We then derive ${}\vdash \subst M N $ by $\regle{@}$,
		setting $\subst M N \defeq (\subst {M'} N)\subst {M''} N$.
\end{itemize}

\begin{rem}\label{linf:rem:mixed}
	The previous construction has the typical structure of the form of reasoning we use in the next sections, and follows the definition of $\linf{001} = \nu Y. \mu X.( \Vv + \lambda\Vv.X + (X)Y )$: it is \enquote{an induction wrapped into a coinduction}.
	
	Although there is no standard notion of \enquote{proof by coinduction} --- at least, one that would be as well established as reasoning by induction --- the only thing we do here is \emph{producing} coinductive objects — namely, derivation trees.
	The derivation trees we produce are \enquote{legal}, since the coinductive steps correspond to occurrences of the coinductive rule $\regle{\coI}$, the syntactic guard being materialised by the later modality $\later$.
	
	Then, each coinductive step is reached by induction from the previous one, which corresponds to the $\mu X$ in $\linf{001}$. This is just a regular induction on the derivation separating two coinductive rules. Notice that this induction has two \enquote{base cases}: when it stops on the rule $\regle{\Vv}$, and when it reaches a coinductive rule $\regle{\coI}$.
	
	This paper is about λ-calculus and not about foundations of reasoning with inductive-coinductive types, so we will forget as much as possible about reasoning technicalities: we keep a lightweight proof style, as classically done for inductive proofs and as described for instance by \autocite{KozenSilva17} and \autocite{Czajka19} for coinductive reasoning.

	In the following, whenever we claim to define some object or to establish some result \enquote{by nested coinduction and induction}, the reader should thus understand that we actually construct some possibly infinite tree (a term or a derivation), following the structure of some input which is itself a possibly infinite tree.
	We then reason by cases on the root of the input tree, assuming the result of the construction is known for immediate subtrees:
	to ensure that this defines an object in the output type, it is sufficient to check that, each time we reach a coinductive step in the input, we proceed with the construction under the guard of at least one coinductive step in the output.
\end{rem}


\begin{defi}[finitary reduction $\flb$]\label{linf:def:flb}
	The relation $\beta_0$ is defined on $\linf{001}$ by:
	\[\beta_0 \defeq \left\{
		\left((\lambda x.M)N,\subst M N \right),\ 
		M,N\in\linf{001}, x\in\mathcal{V}
	\right\}.\]
	
	The relation $\flb$ is then defined on $\linf{001}$ by induction as the \emph{contextual closure} of $\beta_0$, namely:
	\begin{prooftabular}
		\mbox{\begin{prooftree}
			\hypo{ M\ \beta_0\ N }
			\infer1[\ax_{\beta}]{ M\flb N }
		\end{prooftree}}
	&
		\mbox{\begin{prooftree}
			\hypo{ M\flb N }
			\infer1[\lambda_\beta]{ \lambda x.M \flb \lambda x.N }
		\end{prooftree}}
	\ptnewline
		\mbox{\begin{prooftree}
			\hypo{ M\flb N }
			\infer1[@l_\beta]{ (M)P \flb (N)P }
		\end{prooftree}}
	&
		\mbox{\begin{prooftree}
			\hypo{ M\flb N }
			\infer1[@r_\beta]{ (P)M \flb (P)N }
		\end{prooftree}}
	\end{prooftabular}
\end{defi}

\subsection{Infinitary \texorpdfstring{β}{beta}-reduction}

We extend our calculus with an infinitary β-reduction. As already mentioned, an infinite reduction must \enquote{go to infinity}, that is to say that the depth of fired redexes tends to infinity.

\begin{nota}
	Given a relation $\longrightarrow$, we denote $\longrightarrow^?$ its reflexive closure and $\longrightarrow^*$ its reflexive and transitive closure.
\end{nota}

\begin{defi}[001-infinitary reduction $\flbi$]\label{linf:def:flbi}
	The infinitary reduction $\flbi$ is defined on $\linf{001}$ as the 
	\emph{001-strongly convergent closure} of $\flb$, that is to say by the 
	following mixed formal system:
	\begin{prooftabular}[]
		\mbox{\begin{prooftree}
			\hypo{ M \flbs x }
			\infer1[\ax_\beta^\infty]{ M \flbi x }
		\end{prooftree}}
	&
		\mbox{\begin{prooftree}
			\hypo{ M \flbs \lambda x.P }
			\hypo{ P \flbi P' }
			\infer2[\lambda_\beta^\infty]{ M \flbi \lambda x.P' }
		\end{prooftree}}
	\ptnewline
		\mbox{\begin{prooftree}
			\hypo{ M \flbs (P)Q }{depth}
			\hypo{ P \flbi P' }{depth}
			\hypo{ \later Q \flbi Q' }
			\infer3[@_\beta^\infty]{ M \flbi (P')Q' }
		\end{prooftree}}
	&
		\mbox{\begin{prooftree}
			\hypo{ M \flbi M' }
			\infer[double]1[\coI_\beta^\infty]{ \later M \flbi M' }
		\end{prooftree}}
	\end{prooftabular}
\end{defi}

\cref{linf:def:flbi} provides an inductive-coinductive presentation of the notion of strongly convergent reduction sequences defined by \autocite{KennawayAl97}, in the specific setting of $\linf{001}$: the only coinductive step occurs in argument position in the application rule, which is the position where $\mathrm{depth}^{001}$ is incremented.
In that we follow Dal Lago \autocite{DalLago16}, whereas the fully coinductive approach of Endrullis and Polonsky \autocite{EndrullisAl13} is limited to $\linf{111}$.

\begin{exa}
	The well-known $Y = \lambda f.(\Delta_f)\Delta_f$, with $\Delta_f=\lambda x.(f)(x)x$, satisfies $Y\flbi Y^*$. Indeed:
	\begin{flushleft}\begin{prooftree}
		\hypo{ Y \flbs \lambda f.(\Delta_f)\Delta_f }
		\hypo{ (\Delta_f)\Delta_f \flbs (f)(\Delta_f)\Delta_f }
		\hypo{ f \flbs f }
		\infer1{ f \flbi f }
		\hypo{\parbox[t][0pt]{0pt}{
			\begin{tikzpicture}[baseline]
			\draw [-{To[scale=1.3]},dashed] (0,0)
			to[out=90,in=180] (1.3,0.5)
			to[out=0, in=90] (2.6,-0.5)
			to[out=270,in=90] (2.6,-1.3)
			to[out=270,in=0] (0,-2.1);
			\end{tikzpicture}
		}}
		\infer1{ (\Delta_f)\Delta_f \flbi f^\infty }
		\infer[double]1{ \later (\Delta_f)\Delta_f \flbi f^\infty}
		\infer3{ (\Delta_f)\Delta_f \flbi f^\infty=(f)f^\infty}
		\infer2{ Y \flbi Y^* = \lambda f.f^\infty }
	\end{prooftree}\end{flushleft}
\end{exa}

\begin{rem}
	Definitions~\ref{linf:def:flb} and~\ref{linf:def:flbi} could, again, be formulated in terms of fix-points:
	\begin{align*}
	\flb\quad 
		&\defeq \quad \nu Y.\mu X. \left(\beta_0\ +\ \lambda\Vv.X\ +\ (X)\linf{001}\ +\ (\linf{001})Y\right) \\
	\flbi\quad 
		&\defeq \quad \nu Y.\mu X. \left(\flbs\ +\ \flbs \!\relcomp\, 
		(\lambda\Vv.X)\ +\ \flbs \!\relcomp\, (X)Y \right)
	\end{align*}
	where the functors act on relations, for instance 
	$\lambda\Vv.X = 
	\{ (\lambda v.x_1, \lambda v.x_2)\ |\ v\in\Vv, (x_1,x_2)\in X \}$, 
	and the symbol $\relcomp$ denotes the composition of relations.
\end{rem}



\begin{lem}\label{linf:lem:flbi_transitive}
	\hfill
	\begin{enumerate}
	\item $\flbi$ is reflexive.
	\item $\mathord{\flbs}\subseteq\mathord{\flbi}$.
	\item $\flbi$ is transitive.
	\end{enumerate}
\end{lem}

	\begin{proof}
	\begin{enumerate}
	\item For any $M\in\linf{001}$, a derivation of $M\flbi M$ is built 
	straightforwardly by nested coinduction and induction,%
	\footnote{We recall that this proof scheme is discussed in
	\cref{linf:rem:mixed} --- especially in its last paragraph.}
	following the structure of the derivation of ${}\vdash M$.
	
	\item Immediate from the rules of \cref{linf:def:flbi}
	and from the reflexivity of $\flbi$,
	by cases on the reduct of $\flbs$.
	For instance, in the case of an abstraction:
		\begin{center}\begin{prooftree}
			\hypo{ M\flbs \lambda x.P }
			\hypo{ \text{(1)} }
			\infer1{ P \flbi P }
			\infer2{ M \flbi \lambda x.P }
		\end{prooftree}\end{center}
	
	\item To prove transitivity, we have to show a series of sublemmas:
		\begin{align}
		\text{if } M \flbs M',
			&\text{ then } \subst M N \flbs \subst {M'} N 		\tag{i}		\\
		\text{if } M\flbs M'\flbi M'',
			&\text{ then } M\flbi M'' 							\tag{ii}	\\
		\text{if } M \flbi M' \text{ and } N\flbi N',
			&\text{ then } \subst M N \flbi \subst {M'} {N'}	\tag{iii}	\\
		\text{if } M \flbi M'\flb M'',
			&\text{ then } M\flbi M'' 							\tag{iv}	\\
		\text{if } M \flbi M'\flbs M'',
			&\text{ then } M\flbi M'' 							\tag{v}		\\
		\text{if } M \flbi M'\flbi M'',
			&\text{ then } M\flbi M'' 							\tag{vi}
		\end{align}
		(i) and (ii) are immediate, respectively by nested coinduction and induction on $M$ and by case analysis on $M'\flbi M''$.
		\smallskip
		
		To prove (iii), proceed by nested coinduction and induction on $M\flbi M'$. 
		\begin{itemize}
		\item If $M'\bisim x$ and $M\flbs x$, use (i) to get $\subst M N \flbs 
		\subst x N \bisim N \flbi N'$, and conclude with (ii).
		\item If $M'\bisim y$ and $M\flbs y$, use (i) to get $\subst M N \flbs 
		\subst y N \bisim y \bisim \subst y {N'}$ and conclude with (2).
		\item If $M'\bisim \lambda y.P'$, $M\flbs \lambda y.P$ and $P\flbi P'$, 
		use (i) to get $\subst M N \flbs \lambda y.\subst P N$, use the 
		induction hypothesis to get a derivation $\subst P N\flbi 
		\subst{P'}{N'}$ and conclude with (ii).
		\item If $M'\bisim (P')Q'$, $M\flbs (P)Q$, $P\flbi P'$ and $Q\flbi Q'$, 
		use (i) to get $\subst M N \flbs\left(\subst P N\right)\subst Q N$, get 
		a derivation $\subst P N\flbi \subst{P'}{N'}$ by induction, and build 
		$\subst Q N\flbi \subst{Q'}{N'}$ coinductively using 
		$\regle{\coI_\beta^\infty}$ as a guard. Conclude with (ii).
		\end{itemize}
		\smallskip
		
		To prove (iv), proceed by induction on $M'\flb M''$.
		\begin{itemize}
		\item If $M' \beta_0 M''$, that is $M'\bisim (\lambda x.Q')R'$ and 
		$M''\bisim \subst{Q'}{R'}$, the last rules applied in $M\flbi M'$ are 
		the following:
		\begin{center}\begin{prooftree}
			\hypo{ M\flbs (P)R }
			\hypo{ P \flbs \lambda x.Q }
			\hypo{ \vdots }
			\infer1{ Q\flbi Q' }
			\infer2{ P \flbi \lambda x.Q'}
			\hypo{ \vdots }
			\infer1{ R\flbi R' }
			\infer[double]1{ \later R \flbi R' }
			\infer3{ M \flbi M'\bisim (\lambda x.Q')R' }
		\end{prooftree}\end{center}
		so $M \flbs (P)R \flbs (\lambda x.Q)R \flb \subst Q R \flbi 
		\subst{Q'}{R'} \bisim M''$ using (iii), and we can conclude with (ii).
		\item If $M'\bisim \lambda x.P'$ and $M''\bisim \lambda x.P''$ with 
		$P'\flb P''$, then the last rule applied in $M\flbi M'$ is the 
		following:
		\begin{center}\begin{prooftree}
			\hypo{ M \flbs \lambda x.P }
			\hypo{ \vdots }
			\infer1{ P\flbi P' }
			\infer2{ M \flbi M'\bisim \lambda x.P' }
		\end{prooftree}\end{center}
		By induction, $P\flbi P''$, and apply the same rule to obtain $M\flbi M''$.
		\item The two remaining cases $\regle{@l_\beta}$ and $\regle{@r_\beta}$ are similar to the previous one.
		\end{itemize}
		\smallskip
		
		(v) is obtained from (iv) by an easy induction.
		\smallskip
		
		Finally, we show (vi) by nested coinduction and induction on $M'\flbi M''$.
		\begin{itemize}
		\item If $M''\bisim x$ and $M'\flbs x$, the result is immediate from 
		(v).
		\item If $M''\bisim \lambda x.P''$ with $M'\flbs \lambda x.P'$ and 
		$P'\flbi P''$, then from $M\flbi M'$ and $M'\flbs \lambda x.P'$, use 
		(v) to get $M\flbi \lambda x.P'$. This means that there is a $P$ such 
		that $M\flbs \lambda x.P$ and $P\flbi P'$. By induction, $P\flbi P''$, 
		and we can derive:
		\begin{center}\begin{prooftree}
			\hypo{ M \flbs \lambda x.P }
			\hypo{ P\flbi P'' }
			\infer2{ M \flbi M''\bisim \lambda x.P'' }
		\end{prooftree}\end{center}
		\item If $M'\flbi M''$ is derived by rule $\regle{@_\beta^\infty}$ with premises $M'\flbs (P')Q'$, $P'\flbi P''$ and $\later Q'\flbi Q''$, then:
			from $M\flbi M'$ and $M'\flbs (P')Q'$ we obtain $M\flbi (P')Q'$ using (v);
			in particular we obtain terms $P$ and $Q$ such that $M\flbs (P)Q$, $P\flbi P'$ and $Q\flbi Q'$;
			applying the induction hypothesis to $P'\flbi P''$ yields $P\flbi P''$;
			to derive $M\flbi M''$ by rule $\regle{@_\beta^\infty}$, it remains only to build a derivation 
			of $Q\flbi Q''$ coinductively, under the guard of $\regle{\coI}$.
		\qedhere
		\end{itemize}
	\end{enumerate}
	\end{proof}

\begin{rem} \label{linf:rem:nf}
	A consequence of the reflexivity of $\flbi$
	is that it makes no sense to consider \enquote{β$\infty$-normal forms}.
	Thus, as in the finitary calculus, we will call
	\emph{β-normal forms} the terms that cannot be reduced
	through $\flb$.
	Then, a β-normal form of a term $M$ (for the infinitary λ-calculus)
	is a term $N$ in β-normal form such that $M \flbi N$.
\end{rem}

\section{The Taylor expansion of \texorpdfstring{{\upshape λ}}{lambda}-terms}
\label{taylor}

Introduced by Ehrhard and Regnier as a particular case of the \emph{differential λ-calculus} \autocite{EhrhardRegnier03}, the \emph{resource λ-calculus} \autocite{EhrhardRegnier08} is the target language of the Taylor expansion of finite λ-terms:
a λ-term is translated as a set of resource terms ---
or, in a quantitative setting, as a (possibly infinite) weighted sum of resource terms.

In this section, we extend the definition of Taylor expansion to infinite λ-terms.
Note that this generalisation is very straightforward, and it does not require to extend the target of the translation.
Indeed, Ehrhard and Regnier have defined Taylor expansion not only on finite λ-terms but also on Böhm trees,
and our generalisation boils down to observe that there is already enough ``room'' to accommodate all terms in $\linf{001}$.

\subsection{The resource \texorpdfstring{λ}{lambda}-calculus}

First, let us recall the definition of the resource λ-calculus. A more detailed presentation can be found in \autocite{Vaux19,BarbarossaManzonetto20}.

\begin{defi}[resource λ-terms]
	The set $\lr{}$ of \emph{resource terms} on a set of variables $\Vv$ is 
	defined inductively by:
		\begin{align*}
		\lr{} &\defeq \Vv \ |\ \lambda\Vv.\lr{} \ |\ \rapp{\lr{}}\lr{!} \\
		\lr{!} &\defeq \Mfin(\lr{})
		\end{align*}
	where $\Mfin(X)$ is the set of finite multisets on $X$.
	
	We call \emph{resource monomials} the elements of $\lr{!}$.
\end{defi}

To denote indistinctly $\lr{}$ or $\lr{!}$, we write $\lr{(!)}$. The multisets are denoted $\ms{t}=[t_1,\dots,t_n]$, in an arbitrary order. Union of multisets is denoted multiplicatively, and terms are identified to the corresponding singleton: for example, $s\cdot [t,u]= [u,s,t]$. In particular, the empty multiset is denoted $1$. The cardinality of a multiset $\ms t$ is denoted $\mscard \ms t$.

Let $(2,\vee,\wedge)$ be the semi-ring of boolean values, and $\lrsums[(!)]$ the free $2$-module generated by $\lr{(!)}$.
We denote by capital $S, T$ (resp. $\ms{S}, \ms{T}$) the elements of $\lrsums $ (resp. $\lrsums[!]$).
By construction, an element of $\lrsums[(!)]$ is nothing but a finite set of 
resource terms (resp. monomials), so that we find it more practical to stick to 
the additive notation:
\emph{e.g.}, we will write $s+S$ instead of $\{s\}\cup S$, and we write $0$ for 
the empty set of terms or monomials.
In addition, we extend the constructors of $\lr{(!)}$ to $\lrsums[(!)]$ by linearity:
\[ \lambda x.\sum_i s_i \defeq \sum_i (\lambda x.s_i) \qquad
 \rapp{\sum_i s_i}\sum_j\ms{t_j} \defeq \sum_{i,j}\rapp{s_i}\ms{t_j} \qquad
 \left(\sum_i s_i\right)\cdot\ms{T} \defeq \sum_i s_i\cdot\ms{T}. \]

\begin{rem}
	We work in a \emph{qualitative} setting, where $s+s=s$, in opposition with
    the original \emph{quantitative} setting where the semi-ring $(\NN,+,\times)$
    allows to count occurrences of a resource term (for instance, $s+s=2s$).
	This is similar to what is done by, e.g., Barbarossa and Manzonetto
	\cite{BarbarossaManzonetto20}.
	Our choice is motivated by the fact that, as discussed in the introduction,
	the treatment of infinitary reduction forbids us to consider 
	Taylor expansion with coefficients in an arbitrary semi-ring
	(see \cref{taylor:rem:quantitative} for further details):
	we thus restrict to the qualitative version of Taylor expansion, which sends a
	λ-term to a (possibly infinite) set of resource terms.
	We then find natural to consider a qualitative variant of the resource
	calculus itself as well.
\end{rem}

\begin{defi}[substitution of resource terms]
\label{taylor:def:rsubst}
	If $s\in\lr{}$, $x\in\Vv$ and $\ms{t}=[t_1,\dots,t_n]\in\lr{!}$, we define:
	\[ \rsubst{s}{\ms t} \defeq \left\{\begin{array}{cl}
		\displaystyle\sum_{\sigma\in\sS_n} \subst{s}{t_{\sigma(i)}}[x_i] & \text{if $\mathrm{deg}_x(s)=n$} \\
		0 & \text{otherwise}
	\end{array}\right. \]
	where $\mathrm{deg}_x(s)$ is the number of free occurrences of $x$ in $s$, $x_1,\dots,x_n$ is an arbitrary enumeration of these occurrences, and $\subst{s}{t_{\sigma(i)}}[x_i]$ is the term obtained by formally substituting $t_{\sigma(i)}$ to each corresponding occurrence $x_i$.
\end{defi}

A more fine-grained definition can be found in 
\autocite{EhrhardRegnier03,EhrhardRegnier08},
where substitution is built as the result of a \emph{differentiation} operation:
$\rsubst s {\ms{t}}
\defeq \left(\frac{\partial^n s}{\partial x^n}\cdot\ms{t}\right)[0/x]$.

\begin{defi}[resource reduction]\label{taylor:def:flr}
	The \emph{simple resource reduction} $\flro\ \subset 
	\lr{(!)}\times\lrsums[(!)]$ is the smallest relation such that for every 
	$s$, $x$ and $\ms{t}$, $\rapp{\lambda x.s}\ms{t} \flro \rsubst s {\ms{t}}$ 
	holds, and closed under:
	
	\begin{prooftabular}
		\mbox{\begin{prooftree}
			\hypo{ s \flro S }
			\infer1[\lambda_r]{ \lambda x.s \flro \lambda x.S }
		\end{prooftree}}
	&
		\mbox{\begin{prooftree}
			\hypo{ s \flro S }
			\infer1[@l_r]{ \rapp{s}\ms{t} \flro \rapp{S}\ms{t} }
		\end{prooftree}}
	\ptnewline
		\mbox{\begin{prooftree}
			\hypo{ \ms{t} \flro \ms{T} }
			\infer1[@r_r]{ \rapp{s}\ms{t} \flro \rapp{s}\ms{T} }
		\end{prooftree}}
	&
		\mbox{\begin{prooftree}
			\hypo{ s \flro S }
			\infer1[!_r]{ s\cdot\ms{t} \flro S\cdot\ms{t} }
		\end{prooftree}}
	\end{prooftabular}
	
	This relation is extended to $\flr\ \subset\ \lrsums \times \lrsums$ by the rule:

	\begin{center}\begin{prooftree}
		\hypo{ s_0 \flro T_0 }
		\hypo{ \left(s_i \flror T_i\right)_{i=1}^n }
		\infer2[\Sigma_r]{ \sum_{i=0}^{n} s_i \flr \sum_{i=0}^{n} T_i }
	\end{prooftree}\end{center}
\end{defi}

\begin{rem} \label{taylor:rem:sigma_r}
	Some authors, like \autocite{BarbarossaManzonetto20}, prefer the following alternative:
	\begin{center}\begin{prooftree}
		\hypo{ s \flro S }
		\hypo{ s \notin T }
		\infer2[\Sigma_r']{ s+T \flr S+T }
	\end{prooftree}\end{center}
	
	Both versions define the same normal forms, but do not induce the same dynamics. In particular, $\regle{\Sigma_r'}$ preserves the termination of $\flrs$ even in the qualitative setting, whereas $\regle{\Sigma_r}$ allows to reduce $s$ to $s+S$ whenever $s \flro S$, which obviously prevents termination.
	
	However, the assumption $s \notin T$ in $\regle{\Sigma_r'}$ forbids to
	reduce \emph{contextually} in a sum, meaning that with this rule,
	$S \flrs S'$ and $T\flrs T'$ do not straightforwardly imply
	$S + T \flrs S' + T'$.\footnote{
                We could find no counterexample to this implication, but no proof either.
                Our best effort allowed us to prove that $S \flrs S'$ and $S\cap T = \emptyset$ imply $S+T\flrs S'+T$.
			}
	Hence our choice to use $\regle{\Sigma_r}$ instead: ensuring
	the contextuality of $\flr$ gives rise to a strong confluence
	result, as recalled in \cref{taylor:lem:confl_term_flr}
	--- whereas the reduction defined by $\regle{\Sigma_r'}$ is \enquote{only} confluent.\footnote{
                Using the fact that $\regle{\Sigma_r'}$ is a particular case of $\regle{\Sigma_r}$,
                and the fact that the reduction using $\regle{\Sigma_r}$ is confluent and normalising,
                we do know that each $S\in\lrsums$ reduces to a unique normal form $\nf{r}{S}$, and that $\nf{r}{S+T}=\nf{r}{S}+\nf{r}{T}$,
                whatever rule we choose.
                This is sufficient to obtain confluence, but the proof is thus
                indirect for the version using rule $\regle{\Sigma_r'}$.
                By contrast, the proof sketch in \autocite{BarbarossaManzonetto20} claims to rely on a direct proof of local confluence,
				but it is not given in full: in particular the case of sums is not discussed.
				It can be worked out, but it is not straightforward, again because of the lack of contextuality.
        }
	This technical choice will play a crucial role in the following: see in 
	particular 
	\cref{taylor:rem:flrst} and \cref{taylor:lem:flst_properties}.
\end{rem}

\begin{defi}
	The \emph{size} $|-|$ of resource terms is defined inductively by:
	\begin{align*}
		|x|								&\defeq 1
		\\
		|\lambda x.s|					&\defeq 1 + |s|
		\\
		\left|\rapp{s}\ms{t}\right|		&\defeq |s| + \left|\ms{t}\right|
		\\
		\left|[t_1,\dotsc,t_n]\right|	&\defeq 1+\sum_{i=1}^n |t_i|.
	\end{align*}
	The size of a finite sum $S \in \lrsums{}$ is given by
	$|S| \defeq \max_{s\in S} |s|$. By convention, $|0| \defeq 0$.
\end{defi}

\begin{lem}\label{taylor:lem:size}
	Given $s\in\lr{}$ and $S \in \lrsums{}$,
	if $s \flro S$ then $|S| < |s|$.
\end{lem}

\begin{proof}
	We first show that for $s\in\lr{}$, $x\in\Vv$ and $\ms{t} = [t_1,\dots,t_n] \in\lr{!}$, $\left| s\rapp{\ms{t}/x} \right| < \left| \rapp{\lambda x.s}\ms{t} \right|$.
		\begin{itemize}
		\item If $\deg_x(s) \neq n$, $\left| s\rapp{\ms{t}/x} \right| = |0| = 0$
			and $\left| \rapp{\lambda x.s}\ms{t} \right|\geq 3$.
		\item Otherwise, $\left| s\rapp{\ms{t}/x} \right| = |s| - n + \sum_{i=1}^{n}|t_i|$ and $\left| \rapp{\lambda x.s}\ms{t} \right| = |s| + 2 + \sum_{i=1}^{n}|t_i|$, which leads to the expected inequality.
		\end{itemize}
	We obtain the desired result on $\flro$ by induction.
\end{proof}

\begin{rem}\label{taylor:rem:wn}
	Observe that, as a corollary, we obtain
	$|S'|\leq|S|$ whenever $S\flr S'$ in $\lrsums{}$,
	but this inequality is not strict in general;
	indeed, if $S \flr S'$ and $|T| \geq |S|$
	then $S+T \flr S'+T$ but $|S'+T| = |S+T| = |T|$.

	Nonetheless, the previous Lemma ensures the normalisation of $\flr$,
	using the multiset order on the sizes of elements in a finite sum,
	as we now show.
\end{rem}

\Needspace{5cm}
\begin{lem}[normalisation and confluence of $\flr$] 
\label{taylor:lem:confl_term_flr} \hfill
	\begin{enumerate}
	\item The resource reduction 
		$\mathord{\flr} \subset \lrsums \times \lrsums$
		is weakly normalising.
	\item This reduction is strongly confluent in the following sense: whenever there are $S, T_1, T_2 \in \lrsums$ as below, there is a $U \in \lrsums$ such that:
		
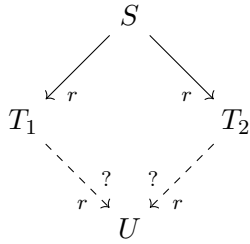
\begin{center}
		\begin{tikzpicture}[node distance=20mm]
		\node (s) { $S$ };
		\node (t1) [below left of=s] { $T_1$ };
		\node (t2) [below right of=s] { $T_2$ };
		\node (u) [below right of=t1] { $U$ };
		\draw[->] (s) -- (t1) 
			node [pos=0.95, right=3pt] { \scriptsize $r$ };
		\draw[->] (s) -- (t2)
			node [pos=0.95, left=3pt] { \scriptsize $r$ };
		\draw[->, dashed] (t1) -- (u)
			node [pos=0.95, left=3pt] { \scriptsize $r$ }
			node [pos=0.95, above=3pt] { \scriptsize $?$ }  ;
		\draw[->, dashed] (t2) -- (u)
			node [pos=0.95, right=3.5pt] { \scriptsize $r$ }
			node [pos=0.95, above=3pt] { \scriptsize $?$ }  ;
		\end{tikzpicture}
		\end{center}
	In particular, it is confluent.
	\end{enumerate}
\end{lem}

\begin{proof}
	We prove (1). For (2), the proof is exactly the same as that of \autocite[Lem.~3.11]{Vaux19}.

	To each $S=\sum_{i=1}^n s_i$ (assuming the $s_i$'s are pairwise distinct),
	we can associate the multiset $\left\Vert S\right\Vert\defeq [|s_1|,\dotsc,|s_n|]$.
	\Cref{taylor:lem:size} entails $\left\Vert S'+T\right\Vert \prec \left\Vert s+T\right\Vert$
	whenever $s\flro S'$ and $s\not\in T$
	--- where $\prec$ denotes the Dershowitz–Manna ordering
	\autocite{DershowitzManna79} induced on $\Mfin(\NN)$ by $<$, which is well-founded.
	This entails the strong normalisation property for the version of $\flr$
	restricted to rule $\regle{\Sigma'_r}$,
	hence the weak normalisation of the more general version we use.
\end{proof}

\begin{nota}
	For $s\in\lr{(!)}$, we write $\nf{r}{s}$ for its normal form.
\end{nota}

\subsection{The Taylor expansion}

Just like the Taylor expansion of a function in calculus, the Taylor expansion 
of a term is a weighted, possibly infinite sum of finite approximants. In our 
qualitative setting, the weights vanish and the Taylor expansion can be seen as 
a mere set of approximants (usually called the \emph{support} of the full 
quantitative Taylor expansion). However, we describe these sets using an 
additive formalism, to be consistent with the finite sums as defined above.

\begin{nota}
	A possibly infinite set $\{ s_i,\ i \in I \} \subset \lr{}$ will be denoted as $\sum_{i \in I} s_i$. In particular, finite sets will be assimilated to the corresponding finite sums in $\lrsums{}$. Accordingly:
	\begin{itemize}
	\item $s \in \sum_{i \in I} s_i$ means that $s$ belongs to the given set (and equivalently, in the finite case, has coefficient $1$ in the given sum),
	\item unions of sets are also denoted with the symbols $+$ and $\sum$,
	\item the singleton $\{s\}$ is assimilated to the 1-term sum $s$.
	\end{itemize}
\end{nota}

In the finitary setting, the definition of the Taylor expansion relies on the following induction:
\[\begin{array}{rcl@{\qquad}rcl}
	\Tay(x) 				&\defeq & x,
&	\Tay(\lambda x.M) 	&\defeq & \sum_{s\in\Tay(M)} \lambda x.s,
\\	\Tay((M)N) 			&\defeq & \sum_{s\in\Tay(M)} \sum_{\ms{t}\in\Tay(N)^!} 
							 \rapp{s}\ms{t},
&	\Tay(M)^! 			&\defeq & \Mfin(\Tay(M)).
\end{array}\]
In our setting, the principle of the definition is exactly the same: to collect 
the finite approximants of an infinitary term, one just has to 
\emph{inductively} scan the term. However, there is no possible 
\enquote{structural induction} on coinductive objects, so that we need to 
define explicitly an approximation relation.

\begin{defi}[Taylor expansion] \label{taylor:def:taylor}
	The relation $\apptay$ of \emph{Taylor approximation} is inductively
    defined on $\lr{} \times \linf{001}$ by:\footnote{
		The relation $\apptay$ could be equivalently defined by induction on
		resource terms, rather than by a system of derivations:
		derivations are actually directed by the syntax of resource terms.
	}
	\begin{center}
		\begin{prooftree}
		\hypo{}
		\infer1[\ax_{\apptay}]{ x \apptay x }
		\end{prooftree}
	\qquad 
		\begin{prooftree}
		\hypo{ s \apptay M }
		\infer1[\lambda_{\apptay}]{ \lambda x.s \apptay \lambda x.M }
		\end{prooftree}
	\quad 
		\begin{prooftree}
		\hypo{ s \apptay M }
		\hypo{ \ms{t} \apptay^! N }
		\infer2[@_{\apptay}]{ \rapp{s}\ms{t} \apptay (M)N }
		\end{prooftree}
	\quad
		\begin{prooftree}
		\hypo{ \left( t_i \apptay M \right)_{i=1}^n }
		\infer1[!_{\apptay}]{ [t_1,\dots,t_n] \apptay^! M }
		\end{prooftree}
	\end{center}
	
	The \emph{Taylor expansion} of a term $M\in\linf{001}$ is the set
	$\Tay(M) \defeq \sum_{s\apptay M} s$.
\end{defi}

Again, it is practical to extend $\apptay$ to sums of resource terms:
we write $\sum_i s_i \apptay M$ whenever $\forall i,\ s_i \apptay M$,
so that $\Tay(M)\apptay M$ and $\Tay(M)$ is the greatest set of resource terms with that property.
Note that, due to the shape of rule $(!_{\apptay})$, $\Tay(M)$ is infinite as soon as $M$ contains an application.

\begin{rem}\label{taylor:rem:coef}
	We could very well consider a quantitative version of Taylor expansion:
	it poses no particular problem to define the coefficient of $s$ in
	$\Tay(M)$ by induction on $s$, following the original definition for
	ordinary λ-calculus \cite{EhrhardRegnier06}.
	Establishing a quantitative version of our simulation result, taking
	coefficients into account, is another matter, because the reduction of
	infinite weighted sums of resource terms is not always well defined:
	see \cref{taylor:rem:quantitative} below.
\end{rem}

\subsection{Reducing (possibly infinite) sets of resource terms}

So far, we are only able to reduce finite sums of resource terms
(using $\flr$),
but the Taylor expansion of a term is an infinite sum in general.
The following definition enables us to lift $\flrs$ from $\lr{} \times 
\lrsums $ to $\Pp(\lr{}) \times \Pp(\lr{})$.

\begin{defi}
	Let $X$ be a set, and $\longrightarrow\ \subset\ X \times \freemod{2}{X}$ a 
	relation. We define a reduction $\flt\ \subset\ \Pp(X)\times\Pp(X)$ by 
	stating that $A\flt B$ if we can write:
	\[
		A = \sum_{i\in I} a_i						\quad\text{,}\quad
		B = \sum_{i\in I} B_i						\quad\text{and}\quad
		\forall i\in I,\ a_i \longrightarrow B_i,
	\]
	where $I$ is a (possibly infinite) set of indices and, for each $i\in I$,
	$a_i \in X$ and $B_i \in \freemod{2}{X}$.
\end{defi}

In the following, we consider the relation $\flrst$ on sets of resource terms.

\begin{rem}\label{taylor:rem:flrst}
	As a direct consequence of the definition,
	given two sets $\Ss, \Ss' \in \Pp(\lr{})$,
	we have $\Ss \flrst \Ss'$ whenever we can write:
	\begin{equation}
		\label{taylor:eqn:flrst}
		\Ss' = \sum_{s \in \Ss} S'_s				\quad\text{and}\quad
		\forall s \in \Ss,\ s \flrs S'_s.
	\end{equation}
	Note in particular that the length of each reduction $s \flrs S'_s$
	might depend on $s$, and is not bounded in general, which will be crucial
	in the following.
\end{rem}
	
It turns out that condition \labelcref{taylor:eqn:flrst} is in fact equivalent
to the definition of $\Ss\flrst \Ss'$.

\begin{lem}\label{taylor:lem:flrst}
	We have $\Ss \flrst \Ss'$ iff condition \labelcref{taylor:eqn:flrst} holds.
\end{lem}
\begin{proof}
	Due to the way we defined $\flr$ from $\flro$, 
	we have $S\flrs S'$ iff we can write
	$S=\sum_{i=1}^n s_i$ and $S'=\sum_{i=1}^n S'_i$
	with $s_i\flrs S'_i$ for $1\leq i\leq n$.
	In particular $s\flrs S'$ iff we can write $S'=\sum_{i=1}^n S'_i$
	with $n>0$ and $s\flrs S'_i$ for $1\leq i\leq n$
	--- note that this latter equivalence holds only because we consider finite
	sets of terms rather than formal sums, so that $s=\sum_{i=1}^n s$.

	Hence condition \labelcref{taylor:eqn:flrst} holds iff we can write
	$\Ss=\sum_{i\in I} s_i$ and $\Ss'=\sum_{i\in I} S'_i$
	so that for each $i\in I$, $s_i\flrs S'_i$,
	and furthermore for each $s\in\lr{}$, $\{i\in I,\ s_i=s\}$ is finite.
	But this finiteness condition is always fulfilled:
	\cref{taylor:lem:size} entails that $|S'|\le|s|$ whenever $s\flrs S'$,
	and since moreover the free variables of $S'$ are also free in $s$,
	we obtain that $\{S',\ s\flrs S'\}$ is finite.
\end{proof}

Observe that several steps in the previous proof rely implicitly on the
contextuality of $\flrs$, obtained thanks to rule $\regle{\Sigma_r}$, as
stressed in \cref{taylor:rem:sigma_r}.

\begin{rem}\label{taylor:rem:quantitative}
	The proof of \cref{taylor:lem:flrst} also shows that $\flrst$ is in fact a variant of the
	relation on (possibly infinite) linear combinations of resource terms
	introduced by the second author in a quantitative setting, in order to simulate
	the β-reduction of ordinary λ-terms
	\autocites{Vaux17}[Def.~5.4]{Vaux19}.

	The only difference is that the underlying reduction on resource terms is 
	the iterated reduction $\flrs$ rather than a parallel variant of $\flr$.
	Indeed, we want to simulate $\flbi$, which amounts to a possibly infinite
	sequence of β-reductions: as discussed in the introduction,
	we cannot bound the number of reductions to be performed on resource terms
	in the simulation of one step of $\flbi$, and we are forced to consider the
	reflexive and transitive closure of resource reduction instead
	--- the only point of parallel resource reduction was precisely to avoid
	the need to consider $\flrs$ in the simulation of a single β-reduction step.

	In particular, we do face obstacles to considering a quantitative
	version of reduction, as studied in \autocites{Vaux19}.
	For instance, observe that the Taylor expansion of the $001$-infinitary
	term $(\lambda x.x)^{\infty}$ contains each resource term of the shape
	\[
		\rapp{\lambda x.x}^k[s] \defeq \underbrace{
		\rapp{\lambda x.x}[\cdots[\rapp{\lambda x.x}[s]]\cdots]
		}_{\text{$k$ nested linear applications}}
	\]
	where $s$ is itself any fixed approximant of  $(\lambda x.x)^{\infty}$.
	Each $\rapp{\lambda x.x}^k[s]$ reduces to $s$ (in $k$ steps):
	in particular, $\sum_{k\in\NN}\rapp{\lambda x.x}^k[s]\flrst \sum_{k\in\NN}s$.
	If we were to take coefficients into account we would have to deal
	with infinite sums of coefficients.
	This is precisely why we stick to the qualitative setting.
\end{rem}

\begin{lem}\label{taylor:lem:flst_properties}\hfill
	\begin{enumerate}
	\item $\flrst$ is reflexive and transitive.
	\item $(\mathord{\flrt})^* \subseteq \mathord{\flrst}$.
	\end{enumerate}
\end{lem}
	
	\begin{proof}
	\begin{enumerate}
	\item Reflexivity is immediate: $A=\sum_{a\in A}\{a\}$ with $a \flrs a$. For transitivity, consider $A\flrst B\flrst C$, that is $B=\sum_{a\in A} B_a$ with $a \flrs B_a$, and $C=\sum_{b\in B} C_b$ with $b \flrs C_b$. From the latter, we have $B_a \flrs \sum_{b\in B_a} C_b$ for each $a$.
		Finally:
			\[ C = \adjustlimits \sum_{a\in A} \sum_{b\in B_a} C_b
			\quad\text{and}\quad
			\forall a\in A,\ a \flrs B_a \flrs \sum_{b\in B_a} C_b, \]
			that is $A\flrst C$.
	
	\item We have $\mathord{\flr} \subset \mathord{\flrs}$,
		from which we deduce $\mathord{\flrt} \subset \mathord{\flrst}$, 
		and finally $(\mathord{\flrt})^* \subseteq 
		\left(\mathord{\flrst}\right)^* = \mathord{\flrst}$ from (1).
	\qedhere
	\end{enumerate}
	\end{proof}

\begin{nota} \label{taylor:nota:nft}
	For every $\Ss\in\Pp(\lr{(!)})$, we write its normal form $\nft{r}{\Ss} = \sum_{s\in \Ss} \nf{r}{s}$.
	
	Observe that $\Ss\flrst\nft{r}{\Ss}$, because $s\flrs\nf{r}{s}$ for each $s\in\Ss$.
\end{nota}

\section{Simulating the infinitary reduction} \label{simul}

The goal of this part is to simulate the infinitary reduction through the Taylor expansion, that is to obtain the following result:
\begin{center}
	if $M \flbi N$, then $\Tay(M) \flrst \Tay(N)$.
\end{center}
We first show that the result holds if $M\flbs N$ 
(\cref{simul:lem:simul_flbs}). Then we decompose $\flbi$ into finite 
\enquote{min-depth} steps $\flbgs{d}$ followed by an infinite $\flbig{d}$ 
(\cref{simul:lem:niveaux}), and we refine this decomposition into a tree of 
(min-depth resource) reductions using the Taylor expansion 
(\cref{simul:cor:niveaux_split}). Finally, after having related the \emph{size} 
and \emph{height} of resource terms, we conclude with a diagonal argument that 
enables us to \enquote{skip} the part related to $\flbig{d}$ in each branch of 
the aforementioned tree (\cref{simul:thm:simul_flbi}).

\subsection{Simulation of the finite reductions}

As a first step, we want to simulate substitution and finite β-reduction 
through the Taylor expansion.
This follows a well-known path, similar to the finitary calculus
\autocite{Vaux17}.

\begin{lem}[simulation of the substitution] \label{simul:lem:simul_subst}
	Let $M,N\in\linf{001}$ be terms, and $x\in\Vv$ be a variable. Then:
	\[ \Tay\left(\subst M N \right)
	= \adjustlimits \sum_{s\in\Tay(M)} \sum_{\ \ms{t}\in\Tay(N)^!} \rsubst s {\ms{t}}. \]
\end{lem}

	\begin{proof}
	We proceed by double inclusion.
	First, we show that for all derivation $u \apptay \subst M N$,
	there exist derivations $s \apptay M$ and $\overline{t} \apptay^! N$ such 
	that $u \in \rsubst s {\ms t}$.
	We do so by induction on $u \apptay \subst M N$, considering the possible cases for $M$:
	\begin{itemize}
	\item If $M=x$ then $\subst MN=N$, hence $u\apptay N$.
		Then we can set $s \defeq x$ and $\ms t \defeq [u]$.
	\item If $M=y\not=x$ then $\subst MN=y$, hence $u=y$.
		Then we can set $s \defeq y$ and $\ms t \defeq 1$.
	\item If $M=\lambda y.M'$ then $\subst MN=\lambda y.\subst {M'}N$,
		hence we must have $u=\lambda y.u'$ with $u'\apptay \subst{M'}N$.
		The induction hypothesis yields $s'\apptay M'$ and $\ms t\apptay^! N$
		such that $u'\in\rsubst {s'}{\ms t}$.
		Then we can set $s \defeq \lambda y.s'$.
	\item If $M=(M')N'$ then $\subst MN=(\subst {M'}N)\subst{M''}N$,
		hence we must have $u=\rapp{u'}{\ms u''}$ with $u'\apptay \subst{M'}N$
		and $\ms u''\apptay^! \subst{M''}N$.
		Writing $\ms u''=[u''_1,\dotsc,u''_n]$,
		this means $u''_i\apptay \subst{M''}N$ for $1\le i\le n$.
		The induction hypothesis applied to $u'\apptay \subst{M'}N$ yields $s'\apptay M'$ and $\ms t_0\apptay^! N$
		such that $u'\in\rsubst {s'}{\ms t_0}$.
		The induction hypothesis applied to each $u''_i\apptay \subst{M''}N$ yields $s''_i\apptay M''$ and $\ms t_i\apptay^! N$
		such that $u''_i\in\rsubst {s''_i}{\ms t_i}$.
		Then we can set $s\defeq\rapp{s'}{[s''_1,\dotsc,s''_n]}$
		and $\ms t\defeq\ms t_0\cdot\cdots\cdot\ms t_n$.
	\end{itemize}
	
	This concludes the first inclusion.
	Conversely, let us show that for all derivations $s \apptay M$ and $\ms t \apptay^! N$,
	$\rsubst{s}{\ms t} \apptay \subst M N$.
	We proceed by induction on the derivation $s \apptay M$:
	\begin{itemize}
	\item If $s \bisim x \apptay x \bisim M$ and $\ms t \bisim [t] \apptay^! 
	N$, then $\rsubst s {\ms t} \bisim t \apptay N \bisim \subst M N$.
	
	\item If $s \bisim y \apptay y \bisim M$ and $\ms t \bisim 1 \apptay^! N$, 
	then $\rsubst s {\ms t} \bisim y \apptay y \bisim \subst M N$.

	\item If $s$ is a variable, but none of the previous two cases apply,
	then we have $\rsubst s{\ms t}=0\apptay \subst MN$ for free.
	
	\item If $s \bisim \lambda y.s' \bisim \lambda y.P \bisim M$ with $s' 
	\apptay P$, then by induction for any derivation $\ms t \apptay^! N$, 
	we have $\rsubst{s'}{\ms t} \apptay \subst P N$. 
	Applying the rule $\lambda_{\apptay}$ to each summand of 
	$\rsubst{s'}{\ms t}$ gives $\rsubst s {\ms t} \bisim \lambda 
	y.\rsubst{s'}{\ms t} \apptay \lambda y.\subst P N \bisim \subst M N$.
	
	\item If $s \bisim \rapp{s'} \ms s'' \apptay (P)Q \bisim M$ with
	$s' \apptay P$ and $\ms s'' \bisim [s''_1,\dots,s''_m] \apptay^! Q$.
	Fix $\ms t \apptay^! N$.
	For each $v\in\rsubst s{\ms t}$,
	there are monomials $\ms t_0,\dotsc,\ms t_m$ such that $\ms t=\ms t_0\cdot\cdots\cdot \ms t_m$
	and $v\in\rapp{\rsubst{s'}{\ms t_0}}[\rsubst{s''_1}{\ms t_1},\dotsc,\rsubst{s''_m}{\ms t_m}]$.
	Observing that $\ms t_i\apptay^! N$ for $0\le i\le m$,
	the induction hypothesis yields:
		\begin{itemize}
		\item $\rsubst{s'}{\ms t_0} \apptay \subst P N$,
		\item and, for $1\le i\le m$, $\rsubst{s''_i}{\ms t_i} \apptay \subst Q N$.
		\end{itemize}
	We obtain $v\apptay (\subst PN)\subst QN \bisim \subst MN$
	by applying the rules $\regle{@_{\apptay}}$ and $\regle{!_{\apptay}}$.
	Hence $\rsubst{s}{\ms t}\apptay \subst{M}{N}$.
	\qedhere
	\end{itemize}
	\end{proof}

\begin{lem}[simulation of the finitary reduction] \label{simul:lem:simul_flbs}
	If $M\flbs N$, then $\Tay(M)\flrst \Tay(N)$.
\end{lem}

	\begin{proof}
	We first show the result for $M\flb N$, by induction on the corresponding derivation.
	
	\begin{itemize}
	\item Case $\regle{\ax_{\beta}}$. Take $M \bisim (\lambda x.P)Q 
	\mathrel{\beta_0} \subst P Q \bisim N$, then:
		\[
			\Tay(M)
		=	\Tay((\lambda x.P)Q) 
		= 	\adjustlimits \sum_{s\in\Tay(\lambda x.P)}
			\sum_{\ \ms{t}\in\Tay(Q)^!} \rapp{s}\ms{t}
		= 	\adjustlimits \sum_{s\in\Tay(P)} \sum_{\ \ms{t}\in\Tay(Q)^!}
			\rapp{\lambda x.s}\ms{t}.
		\]
		Since $\rapp{\lambda x.s}\ms{t} \flr s\rapp{\ms{t}/x}$, we obtain from \cref{simul:lem:simul_subst}:
		\[
				\Tay(M)
		\flrt 	\adjustlimits \sum_{s\in\Tay(P)} \sum_{\ \ms{t}\in\Tay(Q)^!}
				\rsubst s {\ms t}
		= 		\Tay\left( \subst P Q \right) 
		=		\Tay(N).
		\]
		
	\item Case $\regle{\lambda_\beta}$. Take $M \bisim \lambda x.P \flb \lambda 
	x.P' \bisim N$, with $P \flb P'$. By induction, we have $\Tay(P) \flrst 
	\Tay(P')$, that is $\Tay(P') = \sum_{s\in\Tay(P)} S'_s$ with $s\flrs S'_s$. 
	Then:
		\[	\Tay(M)
		=	\smashoperator{\sum_{s \in \Tay(P)}} \lambda x.s
		\qquad
		\text{and}
		\qquad
			\Tay(N)
		= 	\smashoperator{\sum_{s' \in \Tay(P')}} \lambda x.s'
		= 	\smashoperator{\sum_{s \in \Tay(P)}} \lambda x.S'_s,
		\]
	with $\lambda x.s \flrs \lambda x.S'_s$, so $\Tay(M) \flrst \Tay(N)$.
		
	\item Case $\regle{@l_\beta}$: similar to the previous one.
	
	\item Case $\regle{@r_\beta}$. Take $M \bisim (P)Q \flb (P)Q' \bisim N$, 
	with $Q \flb Q'$. By induction, we have $\Tay(Q) \flrst \Tay(Q')$, that is 
	$\Tay(Q') = \sum_{t\in\Tay(Q)} T'_t$ with $t \flrs T'_t$. Then:
		\[	\Tay(M)
		=	\adjustlimits \sum_{s \in \Tay(P)} \sum_{\ms t \in \Tay(Q)^!}
			\rapp{s} \ms t
		\]
		and:
		\begin{align*}
			\Tay(N)
		& =	\adjustlimits \sum_{s \in \Tay(P)} \sum_{\ms t' \in \Tay(Q')^!}
			\rapp{s} \ms t' \\
		& =	\sum_{s \in \Tay(P)} \sum_{k\in\NN}
			\sum_{t'_1 \in \Tay(Q')} \dots \sum_{t'_k \in \Tay(Q')}
			\rapp{s}[t'_1, \dots, t'_k] \\
		& =	\sum_{s \in \Tay(P)} \sum_{k\in\NN}
			\sum_{t_1 \in \Tay(Q)} \dots \sum_{t_k \in \Tay(Q)}
			\rapp{s}[T'_{t_1}, \dots, T'_{t_k}] \\
		& =	\adjustlimits \sum_{s \in \Tay(P)} \sum_{\ms t \in \Tay(Q)^!}
			\rapp{s} [T'_{t_1}, \dots, T'_{t_{\mscard \ms t}}]
		\end{align*}
		Yet for any $\ms t \in \Tay(Q)$ and for all $t_i \in \ms t$, $t_i \flrs T'_{t_i}$, so $\rapp{s} \ms t \flrs \rapp{s} [T'_{t_1}, \dots, T'_{t_{\mscard \ms t}}]$. This leads to $\Tay(M) \flrst \Tay(N)$.
	\end{itemize}
	
	We conclude in the general case $M\flbs N$ using \cref{taylor:lem:flst_properties}.
	\end{proof}

\subsection{A \enquote{step-by-step} decomposition of the reduction}

Since an infinitary reduction must reduce redexes 
whose depth tends to infinity, 
we want to decompose reductions into an infinite succession of 
finite reductions occuring at lower-bounded depth
(in the following, we name these \emph{min-depth} reductions,
for reductions occuring \enquote{at a minimal depth}). 
As a side consequence, we obtain a result of (weak) standardisation 
for $\linf{001}$.

\begin{defi}[min-depth finitary β-reduction]
	The reduction $\flbg{d}\ \subset\ \linf{001}\times\linf{001}$ is defined for all $d\in\NN$ by the rules:
	\begin{prooftabular}
		\mbox{\begin{prooftree}
			\hypo{ M \flb N }
			\infer1[\ax_{\beta\geq 0}]{ M \flbg{0} N }
		\end{prooftree}}
	&
		\mbox{\begin{prooftree}
			\hypo{ M \flbg{d+1} N }
			\infer1[\lambda_{\beta\geq d+1}]{ \lambda x.M \flbg{d+1} \lambda 
			x.N }
		\end{prooftree}}
	\ptnewline
		\mbox{\begin{prooftree}
			\hypo{ M \flbg{d+1} N }
			\infer1[@l_{\beta\geq d+1}]{ (M)P \flbg{d+1} (N)P }
		\end{prooftree}}
	&
		\mbox{\begin{prooftree}
			\hypo{ M \flbg{d} N }
			\infer1[@r_{\beta\geq d+1}]{ (P)M \flbg{d+1} (P)N }
		\end{prooftree}}
	\end{prooftabular}
\end{defi}

\begin{rem}\label{simul:rem:flbgs}
	It is easy to check that if $M\flbg{d} N$ then $M\flb N$, and $M\flbg{d'} 
	N$ whenever $d\geq d'$.
	Moreover, one can show by induction on $d$ that $\flbgs{d}$ can be defined
	directly by the following rules:
	\begin{prooftabular}
		\mbox{\begin{prooftree}[center=false]
			\hypo{ M \flbs M' }
			\infer1[\ax_{\beta\geq 0}^*]{ M \flbgs{0} M' }
		\end{prooftree}}
	&
		\mbox{\begin{prooftree}[center=false]
			\infer0[\Vv_{\beta\geq d+1}^*]{ x \flbgs{d+1} x }
		\end{prooftree}}
	\ptnewline
		\mbox{\begin{prooftree}
			\hypo{ M \flbgs{d+1} M' }
			\infer1[\lambda_{\beta\geq d+1}^*]{ \lambda x.M' \flbgs{d+1} 
			\lambda x.M' }
		\end{prooftree}}
	&
		\mbox{\begin{prooftree}
			\hypo{ M \flbgs{d+1} M' }
			\hypo{ N \flbgs{d} N' }
			\infer2[@_{\beta\geq d+1}^*]{ (M)N \flbgs{d+1} (M')N' }
		\end{prooftree}}
	\end{prooftabular}
	and the infinitary version to be defined below will follow the same pattern.
\end{rem}

\begin{defi}[min-depth resource reduction]
	The reduction $\flrgo{d}\ \subset\ \lr{(!)}\times \lrsums[(!)]$ is defined 
	for all $d\in\NN$ by the rules:
	\vskip 2ex
	\begin{center}
		\begin{prooftree}
			\hypo{ s \flro S }
			\infer1[\ax_{r\geq 0}]{ s \flrgo{0} S }
		\end{prooftree}
		\quad
		\begin{prooftree}
			\hypo{ s \flrgo{d+1} S }
			\infer1[\lambda_{r\geq d+1}]{ \lambda x.s \flrgo{d+1} \lambda x.S }
		\end{prooftree}
		\quad	
		\begin{prooftree}
			\hypo{ s \flrgo{d+1} S }
			\infer1[@l_{r\geq d+1}]{ \rapp{s}\ms{t} \flrgo{d+1} \rapp{S}\ms{t} }
		\end{prooftree}
		\vspace{1.5ex}
		
		\begin{prooftree}
			\hypo{ \ms{t} \flrgo{d+1} \ms{T} }
			\infer1[@r_{r\geq d+1}]{ \rapp{s}\ms{t} \flrgo{d+1} \rapp{s}\ms{T} }
		\end{prooftree}
		\qquad
		\begin{prooftree}
			\hypo{ s \flrgo{d} S }
			\infer1[!_{r\geq d+1}]{ s\cdot\ms{t} \flrgo{d+1} S\cdot\ms{t} }
		\end{prooftree}
	\end{center}
	\vskip 2ex
    where $d\in\NN$. We also extend $\flrgo{d}$ to 
    $\mathord{\flrg{d}}\subset \lrsums{} \times \lrsums{}$ in the same way 
    as in \cref{taylor:def:flr} by adding a rule $\regle{\Sigma_{r\geq d}}$.
\end{defi}

\begin{lem}[simulation of min-depth finitary reduction]
\label{simul:lem:simul_flbs_mindepth}
	Let $M,N \in\linf{001}$ be terms, and $d\in\NN$. If $M\flbgs{d} N$, then $\Tay(M)\flrgst{d} \Tay(N)$.
\end{lem}

	\begin{proof}
		By induction on $\flbgs{d}$. In the case $\regle{\ax_{\beta\geq 0}}$, 
		just apply \cref{simul:lem:simul_flbs}. In the other cases, the proof 
		is analogous to the corresponding cases in \cref{simul:lem:simul_flbs}.
	\end{proof}

\begin{defi}[min-depth infinitary β-reduction]
	The reduction $\flbig{d}$ is defined for all $d\in\NN$ by the rules:
	\begin{prooftabular}
		\mbox{\begin{prooftree}[center=false]
			\hypo{ M \flbi M' }
			\infer1[\ax_{\beta\geq 0}^\infty]{ M \flbig{0} M' }
		\end{prooftree}}
	&
		\mbox{\begin{prooftree}[center=false]
			\infer0[\Vv_{\beta\geq d+1}^\infty]{ x \flbig{d+1} x }
		\end{prooftree}}
	\ptnewline
		\mbox{\begin{prooftree}
			\hypo{ M \flbig{d+1} M' }
			\infer1[\lambda_{\beta\geq d+1}^\infty]{ \lambda x.M' \flbig{d+1} 
			\lambda x.M' }
		\end{prooftree}}
	&
		\mbox{\begin{prooftree}
			\hypo{ M \flbig{d+1} M' }
			\hypo{ N \flbig{d} N' }
			\infer2[@_{\beta\geq d+1}^\infty]{ (M)N \flbig{d+1} (M')N' }
		\end{prooftree}}
	\end{prooftabular}
	where $d\in\NN^*$.
\end{defi}

\begin{lem}\label{simul:lem:flbig}
	Each relation $\flbig{d}$ is reflexive and transitive, and moreover
	such that $\mathord{\flbig{d+1}}\subseteq\mathord{\flbig{d}}\subseteq\mathord{\flbi}$.
\end{lem}
	\begin{proof}
	Reflexivity and transitivity of $\flbig{0}$ are derived from that of $\flbi$.
	Reflexivity of each $\flbig{d}$ follows by induction on $d$: in the
	inductive step, we reason by induction on the top-level (inductive
	only) structure of terms.
	Transitivity of each $\flbig{d}$ follows, again by a straightforward induction on $d$,
	and by induction on derivations of $\flbig{d}$.
	The identity $\mathord{\flbig{0}}=\mathord{\flbi}$ is straightforward.
	The inclusion $\mathord{\flbig{1}}\subseteq\mathord{\flbi}$
	is proved by induction on the derivations of $\flbig{1}$,
	using the reflexivity and transitivity of $\flbi$.
	The inclusion $\mathord{\flbig{d+1}}\subseteq\mathord{\flbig{d}}$
	follows by induction on $d$, and then by induction on the derivations
	of $\flbig{d+1}$.
	\end{proof}

\Needspace{2\baselineskip}
\begin{lem}\label{simul:lem:niveau1}
	If $M\flbi M'$, then there exists a term $M_0\in\linf{001}$ 
	such that $M\flbs M_0\flbig{1} M'$.
\end{lem}

\begin{proof}
	We define $M_0$ and the reduction $M_0\flbig{1} M'$ 
	by reasoning on the inductive layer of the reduction
	$M\flbi M'$.
		\begin{itemize}
		\item Case $\regle{\ax_{\beta}^\infty}$,
			$M\flbs x=M'$. We can set $M_0\defeq x\flbig{1} M'$,
			using the reflexivity of $\flbig{1}$.
		\item Case $\regle{\lambda_{\beta}^\infty}$,
			$M\flbs \lambda x.P$ and $M'=\lambda x.P'$ with $P\flbi P'$.
			The induction hypothesis yields $P_0$
			such that $P\flbs P_0\flbig{1} P'$,
			and we can set $M_0\defeq \lambda x.P_0$
			and obtain $M_0\flbig{1} M'$ by rule $\regle{\lambda_{\beta\geq 
			1}^\infty}$.
		\item Case $\regle{@_{\beta}^\infty}$,
			$M\flbs (P)Q$ and $M'=P'Q'$ with $P\flbi P'$ and $Q\flbi Q'$.
			The induction hypothesis applies to the first reduction, which
			yields $P_0$ such that $P\flbs P_0\flbig{1} P'$.
			Rule  $\regle{\ax_{\beta\geq 0}^\infty}$
			entails $Q\flbig{0} Q'$.
			We can set $M_0\defeq (P_0)Q$, and obtain
			$M_0\flbig{1} M'$ by rule $\regle{@_{\beta\geq 1}^\infty}$.
			\qedhere
		\end{itemize}
\end{proof}

\begin{lem}\label{simul:lem:niveaud1}
	If $M\flbig{d} M'$, then there exists a term $M_d\in\linf{001}$ 
	such that $M\flbgs{d} M_d\flbig{d+1} M'$.
\end{lem}
\begin{proof}
	We define $M_d$ and the reductions $M\flbgs{d} M_d\flbig{d+1} M'$
	by induction on the reduction $M\flbig{d} M'$.
	The case of rule $\regle{\ax_{\beta\geq 0}^\infty}$
	is given by the previous Lemma.
	All the other cases are straightforward using the induction hypothesis.
\end{proof}

\begin{lem}\label{simul:lem:niveaux}
	For all $M,N \in\linf{001}$ such that $M\flbi N$,
	there is a sequence of terms $(M_d)_{d\in\NN}$
	such that for all $d\in\NN$:
	\[			M=M_0
	\flbgs{0}	M_1
	\flbgs{1} 	M_2 
	\flbgs{2} 	\dots 
	\flbgs{d-1} M_d 
	\flbig{d}	N.
	\]
\end{lem}

\begin{proof}
	We construct the sequence $(M_d)_{d\in\NN}$ inductively,
	by applying the previous Lemma.
\end{proof}

\begin{rem}[standardisation for $\flbi$]
	The decomposition of \cref{simul:lem:niveaux} can be slightly improved:
	if $M\flbi N$, there exist $M_0,M_1, M_2, \dots \in\linf{001}$ such that, for all $d\in\NN$:
	\[ M=M_0 \flbes{0} M_1 \flbes{1} M_2 \flbes{2} \dots \flbes{d-1} M_d \flbig{d} N \]
	where $\flbe{d}$ is defined as expected.
	This consequence is a weak counterpart to Curry and Feys' standardisation theorem for the λ-calculus \autocite{CurryFeys58}.
	Another, similar, standardisation theorem has been proved for $\linf{111}$ by Endrullis and Polonsky, using coinductive techniques \autocite{EndrullisAl13}.
\end{rem}

\begin{proof}[Sketch of proof]
	Using a classical result \autocite[lemma~11.4.6]{Barendregt84} which can be easily adapted to $\linf{001}$,
	$M\flbs N$ can be decomposed into \emph{head} and \emph{internal} reductions: $M\flhs M' \longrightarrow_i^* N$.
	Using this, one can prove by induction on $N$ that $M \flbes{0} M_1 \flbgs{1} N$.
	Indeed, we can write either
	\[N=\lambda x_1\dots x_m.(\dots((y)Q_1)\dots)Q_n\]
	in case it is a head normal form, or
	\[N=\lambda x_1\dots x_m.(\dots((\lambda z.P)Q_0)\dots)Q_n\]
	if it has a head redex;
	and then by the definition of internal reduction we must have respectively
	\[M'=\lambda x_1\dots x_m.(\dots((y)Q'_1)\dots)Q'_n\]
	or
	\[M'=\lambda x_1\dots x_m.(\dots((\lambda z.P')Q'_0)\dots)Q'_n\]
	with $Q'_i\flbs Q_i$ for each $i$,
	and also $P'\flbs P$ in the second case.
	In the case of a head normal form, we obtain $M'\flbgs{1}N$ so we can set $M_1\defeq M'$ directly.
	In the other case, we use the induction hypothesis on $P'$,
	to obtain $P_1$ such that $P'\flbes{0}P_1\flbgs{1}N$, and then set
	\[M_1\defeq \lambda x_1\dots x_m.(\dots((\lambda z.P_1)Q'_0)\dots)Q'_n\]
	so that $M\flbes{0}M'\flbes{0}M_1\flbgs{1} N$.

	Then one deduces that $M\flbgs{d} N$ implies $M\flbes{d}M_d\flbgs{d+1} N$,
	by induction on the derivation of $M\flbgs{d} N$.
	Standardisation follows by applying this result to the 
	sequence of reductions obtained by \cref{simul:lem:niveaux}.
\end{proof}
We do not detail the proof (nor the definition of $\flbes{d}$) further, because
this standardisation result is not used in the following:
at this point of the paper, the interested reader already has all the tools to
complete the construction.

\subsection{Decomposing the decomposition}

Each finite, min-depth reduction occuring in the decomposition of 
\cref{simul:lem:niveaux} can be simulated by the Taylor expansion. Using this 
fact, we can track the successive reducts of each approximant in the Taylor 
expansion of the original term $M$, providing a decomposition of each 
$\Tay(M_d)$ into finite sums of approximants.

\begin{lem}[additive splitting] \label{simul:lem:add_split}
	Let $\Ss,\Tay \subset \lr{}$ be sets, and $d\in\NN$.
	If $\Ss \flrgst{d} \Tay$ then, whenever we write
	$\Ss=\sum_{i\in I} S_i$ where each $S_i$ is a finite sum,
	there exist finite sums $T_i$ for $i\in I$,
	such that $\Tay = \sum_{i\in I} T_i$ and $\forall i\in I,\ S_i \flrgs{d} T_i$.
\end{lem}

	\begin{proof}
	For each $i\in I$, write $S_i = \sum_{j\in J_i} s_{i,j}$ with $s_{i,j} \in \lr{}$,
	so that $\Ss = \sum_{i\in I} \sum_{j\in J_i} s_{i,j}$.
	Since $\Ss \flrgst{d} \Tay$ and using \cref{taylor:rem:flrst},
	there are finite sums $T_s$ for $s\in\Ss$
	such that $\Tay = \sum_{s\in\Ss} T_s$ and $\forall s\in\Ss,\ s \flrgs{d} T_{s}$.
	Define, for each $i\in I$, $T_i \defeq \sum_{j\in J_i} T_{s_{i,j}}$.
	It is straightforward to prove that for all $i\in I$, $S_i \flrgs{d} T_i$
	(by induction on the sum of the lengths of reductions $s_{i,j} \flrgs{d} T_{s_{i,j}}$ for $j\in J_i$).
	\end{proof}

\begin{cor} \label{simul:cor:niveaux_split}
	Let $M, N \in \linf{001}$ be terms such that $M \flbi N$,
	and $(M_d)_{d\in\NN}$ given by \cref{simul:lem:niveaux}.
	If $\Tay(M) = \sum_{i\in I} s_i$, then for each $d\in\NN$
	there exist finite sums $(T_{d,i})_{i\in I}$ such that:
	\begin{enumerate}
		\item $\forall i\in I,\ T_{0,i} = s_i$,
		\item $\forall d\in \NN,\ \Tay(M_d) = \sum_{i\in I} T_{d,i}$,
		\item $\forall d\in \NN,\ \forall i\in I,\ T_{d,i} \flrgs{d} T_{d+1,i}$.
	\end{enumerate}
\end{cor}

	\begin{proof}
	For each $i\in I$, set $T_{0,i} \defeq s_i$ and define $T_{d,i}$ by induction on $d$ using the previous lemma and the fact that $\Tay(M_d) \flrgst{d} \Tay(M_{d+1})$, which is a consequence of \cref{simul:lem:simul_flbs_mindepth}.
	\end{proof}

\subsection{Height of a resource term}

We show a few properties of the interplay between the size and the height (wrt.
the 001-depth) of a term, that will play a crucial role in the main proof.

\begin{defi}[height of resource terms]
	The \emph{height} $h^{001}(\cdot)$ of resource terms
	is defined inductively by:
	\begin{align*}
		h^{001}(x)						&\defeq 0 						\\
		h^{001}(\lambda x.s)			&\defeq h^{001}(s) 			\\
		h^{001}(\rapp{s}\ms{t})	&\defeq 
			\max\left( h^{001}(s) , h^{001}(\ms{t}) \right) \\
		h^{001}([t_1,\dotsc,t_n])			&\defeq 1+\max_{1\leq i\leq n} 
		h^{001}(t_i)
	\end{align*}
	The height of a finite sum $S \in \lrsums{}$ is given by
	$h^{001}(S) \defeq \max_{s\in S} h^{001}(s)$
	--- by convention, $h^{001}(0)=0$.
\end{defi}

\begin{lem} \label{simul:lem:h001_size}
	For all $S\in \lrsums $, $h^{001}(S) \leq |S|$.
\end{lem}

\begin{proof}
	Show the result for $s\in\lr{}$ by an immediate induction on $s$. Conclude by taking the maximum over $s\in S$.
\end{proof}

\begin{lem} \label{simul:lem:h001_flrgs}
	Let $S\in \lrsums $ be a finite sum of resource terms and $d\in\NN$ such that $d>h^{001}(S)$.
	Then there is no reduction $S\flrg{d} S'$.
\end{lem}

\begin{proof}
	The result immediately follows from the fact that
	given $d$, $s$ and $S'$, if $s\flrgo{d}S'$ then $d\leq h^{001}(s)$.
	We prove this by induction on the derivation $s\flrgo{d}S'$.
	\begin{itemize}
	\item Case $\regle{\ax_{r\geq 0}}$, $d=0$ and the result is trivial.
	\item Case $\regle{\lambda_{r\geq d+1}}$,
		$s \bisim \lambda x.u$ and $S'\bisim \lambda x.U'$ with $u\flrg{d+1} U'$.
		We conclude directly by the induction hypothesis since 
		$h^{001}(u) = h^{001}(u)$.
	\item Cases $\regle{@l_{r\geq d+1}}$ and $\regle{@r_{r\geq d+1}}$
		are similar to the previous one.
	\item Cases $\regle{!_{r\geq d+1}}$,
		$\ms s=[t] \cdot \ms u$
		and
		$\ms S'=[T'] \cdot \ms u$
		with $t\flrg{d} T'$.
		By induction hypothesis, we have $d\leq h^{001}(t)$,
		hence $d+1\leq 1+h^{001}(t)\leq \max(1+h^{001}(t),h^{001}(\ms 
		u))=h^{001}(\ms s)$.
	\end{itemize}
\end{proof}

\subsection{The diagonal argument}\label{simul:sub:diag}

Finally, we conclude by a sort of diagonal argument: $\Tay(N)$ is shown to be the union of the $T_{i,d_i}$, each of these finite sums being finitely reached from some $s_i\in\Tay(M)$. In that sense, we obtain a pointwise finitary simulation of the infinitary reduction.

The key definition is somehow complementary to the min-depth reduction: 
it is the Taylor expansion at (upper-)\emph{bounded} depth, 
defined hereunder. 
Concretely, $\Taymax{d}(M)$ is the sum of all approximants $s \in \Tay(M)$ 
such that $h^{001}(s) < d$.

\begin{defi}[bounded-depth Taylor expansion]
	For all $d \in \NN$, the relation $\apptaymax{d}$ of \emph{Taylor 
	approximation at depth bounded by $d$} is inductively defined on $\lr{} 
	\times \linf{001}$ by $\mathord{\apptaymax{0}} \defeq \emptyset$ and by the 
	following rules:
	\begin{prooftabular}
		\mbox{\begin{prooftree}
		\hypo{}
		\infer1[\ax_{\apptay < d+1}]{ x \apptaymax{d+1} x }
		\end{prooftree}}
	& 
		\mbox{\begin{prooftree}
		\hypo{ s \apptaymax{d+1} M }
		\infer1[\lambda_{\apptay < d+1}]{ \lambda x.s \apptaymax{d+1} \lambda x.M }
		\end{prooftree}}
	\ptnewline 
		\mbox{\begin{prooftree}
		\hypo{ s \apptaymax{d+1} M }
		\hypo{ \ms{t} \apptaymax{d+1}^! N }
		\infer2[@_{\apptay < d+1}]{ \rapp{s}\ms{t} \apptaymax{d+1} (M)N }
		\end{prooftree}}
	&
		\mbox{\begin{prooftree}
		\hypo{ \left( t_i \apptaymax{d} M \right)_{i=1}^n }
		\infer1[!_{\apptay < d+1}]{ [t_1,\dots,t_n] \apptaymax{d+1}^! M }
		\end{prooftree}}
	\end{prooftabular}
	
	The Taylor expansion of a term $M\in\linf{001}$ at depth bounded by $d$
	is the set
	$\Taymax{d}(M) \defeq \sum_{s\apptaymax{d} M} s$.
	We also write $\Taymaxbang{d}(M) \defeq \sum_{\ms s\apptaymax{d}^! M} \ms s$.
\end{defi}

It should be clear that we have $s\apptaymax{d} M$ iff $s\apptay M$ and $h^{001}(s)<d$,
so that we obtain the following Lemma.

\begin{lem} \label{simul:lem:h001_taylor}
	Let $M\in\linf{001}$ be a term, $S\in \lrsums $ a finite sum of resource terms, and $d\in\NN$. If $S \subset \Tay(M)$ and $h^{001}(S) < d$, then $S \subset \Taymax{d}(M)$.
\end{lem}

\begin{lem} \label{simul:lem:taylor_maxdepth_eq}
	Let $M,N \in \linf{001}$ be terms.
	If $M\flbig{d} N$ then $\Taymax{d}(M) = \Taymax{d}(N)$.
\end{lem}

	\begin{proof}
	We prove the result by induction on $M \flbig{d} N$.
	\begin{itemize}
	\item Case $\regle{\ax_{\beta\geq 0}^\infty}$,
		$\Taymax{0}(M) = 0 = \Taymax{0}(N)$.
	\item Case $\regle{\Vv_{\beta\geq d+1}^\infty}$,
		$N \bisim x\bisim M$ so $\Taymax{d+1}(M) = x = \Taymax{d+1}(N)$.
	\item Case $\regle{\lambda_{\beta\geq d+1}^\infty}$,
		$M \bisim \lambda x.P \flbig{d+1} \lambda x.P' \bisim N$, with $P \flbig{d+1} P'$.
		By induction, $\Taymax{d+1}(P) = \Taymax{d+1}(P')$ so 
		$\Taymax{d+1}(M) = \Taymax{d+1}(N)$ using the rule
		$\regle{\lambda_{\apptay < d+1}}$.
	\item Case $\regle{@_{\beta\geq d+1}^\infty}$,
		$M \bisim (P)Q \flbig{d+1} (P')Q' \bisim N$, with $P \flbig{d+1} P$ and $Q \flbig{d} Q'$.
		By induction, $\Taymax{d+1}(P) = \Taymax{d+1}(P')$ and 
		$\Taymax{d}(Q) = \Taymax{d}(Q')$,
		so $\Taymaxbang{d+1}(Q) = \Taymaxbang{d+1}(Q')$ by the rule $\regle{!_{\apptay < d+1}}$, and 
		finally $\Taymax{d+1}(M) = \Taymax{d+1}(N)$ by the rule $\regle{@_{\apptay < d+1}}$.
	\end{itemize}
	\end{proof}

\begin{thm}[simulation of the infinitary reduction] \label{simul:thm:simul_flbi}
	Let $M,N \in \linf{001}$ be terms. If $M\flbi N$, then $\Tay(M)\flrst \Tay(N)$.
\end{thm}

	\begin{proof}
	Suppose $M\flbi N$.
	By \cref{simul:lem:niveaux},
	we obtain terms $M_0, M_1, M_2, \ldots \in\linf{001}$ such that, for all $d\in\NN$:
	\[			M=M_0
	\flbgs{0}	M_1
	\flbgs{1} 	M_2 
	\flbgs{2} 	\dots 
	\flbgs{d-1} M_d 
	\flbig{d}	N.
	\]
	Writing $\Tay(M) = \sum_{i\in I} s_i$,
	\cref{simul:cor:niveaux_split} yields finite sums $T_{d,i}$ such that:
	\begin{enumerate}
		\item $\forall i\in I,\ T_{0,i} = s_i$,
		\item $\forall d\in \NN,\ \Tay(M_d) = \sum_{i\in I} T_{d,i}$,
		\item $\forall d\in \NN,\ \forall i\in I,\ T_{d,i} \flrgs{d} T_{d+1,i}$.
	\end{enumerate}
	For $i\in I$, define $d_i \defeq |s_i|+1$ and $T_i \defeq T_{d_i,i}$. Using 
	\cref{simul:lem:h001_size}, for all $d\in\NN$, $h^{001}(T_{d,i}) \leq 
	|T_{d,i}| \leq |s_i|$. Thus, $h^{001}(T_i) < d_i$.
	
	\begin{itemize}
		\item From \cref{simul:lem:h001_taylor} and \cref{simul:lem:taylor_maxdepth_eq}, we have $T_i \subset \Taymax{d_i}(M_{d_i}) = \Taymax{d_i}(N) \subset \Tay(N)$. Hence, $\sum_{i\in I} T_i \subseteq \Tay(N)$.
		
		\item Take $t\in\Tay(N)$. From \cref{simul:lem:h001_taylor}, $t\in 
		\Taymax{h}(N)$ where $h \defeq h^{001}(t)+1$. With 
		\cref{simul:lem:taylor_maxdepth_eq}, $t \in \Taymax{h}(M_h) \subset 
		\Tay(M_h)$, so $\exists i\in I,\ t\in T_{h,i}$. For all $d \geq h$, 
		$T_{h,i} \flrgs{h} T_{d,i}$ so, by \cref{simul:lem:h001_flrgs}, $t\in 
		T_{d,i}$.
		
		Notice that for all $d\geq d_i$, $T_i \flrgs{d_i} T_{d,i}$ so using 
		again \cref{simul:lem:h001_flrgs}, $T_{d,i} = T_i$.
		
		Thus, if we take $d \geq \max(h,d_i)$, we obtain $t\in T_i$. This leads 
		us to $\Tay(N) \subseteq \sum_{i\in I} T_i$.
	\end{itemize}
	
	Finally, $\Tay(M) = \sum_{i\in I} s_i$, $\Tay(N) = \sum_{i\in I} T_i$, and $\forall i\in I,\ s_i \flrs T_i$. This implies the theorem.
	\end{proof}

\begin{rem}[models of $\linf{001}$] \label{simul:rem:models}
	An important consequence of this simulation result is that any model $\Mm$ of $\lr{}$ is also a model of $\linf{001}$,
	as soon as it makes sense to consider infinite sets of resource terms:
	it suffices to interpret any term $M \in \linf{001}$ by $\llbracket M\rrbracket_{\Mm} \defeq \bigoplus_{s\in\Tay(M)} \llbracket s \rrbracket_{\Mm}$.
	This is in particular the case of the well-known construction of a 
	reflexive object $\Dd$ in the category $\MRel$ \autocite{BucciarelliAl07}.
\end{rem}

\section{Head reduction and normalisation properties} \label{norm}

In this final part, we show several consequences of \cref{simul:thm:simul_flbi}. Most of them are \emph{unsurprising}, which is good news:
we want $\linf{001}$ to be a convenient framework to consider reductions and normal forms of the usual λ-terms.
In particular, head- and β-normalisation can be characterised in a similar fashion as in the finitary λ-calculus,
and we prove infinitary counterparts to well-known results such as the Commutation theorem and the Genericity lemma.

\subsection{Solvability in $\linf{001}$}

The definition of $\linf{001}$ is tightly related to head reduction:
the inductive/coinductive structure of $001$-infinitary terms follows the head structure of terms;
and since $\flbg{0}$ contains head reduction, \cref{simul:lem:niveaux} entails that any $\flbi$-reduction
involves only finitely many head reduction steps.

The good properties of head reduction should thus be preserved when moving from $\Lambda$ to $\linf{001}$.
This is indeed the case, as expressed by \cref{norm:thm:charac_head}:
a 001-infinitary term is head-normalising
if and only if the head reduction strategy terminates.
As a consequence, we will show that the notion of solvability is completely preserved in $\linf{001}$.

\begin{lem}[head forms] \label{norm:lem:head}
	Let $M\in\linf{001}$ be a term, then either
	\[ M \quad\bisim\quad \lambda x_1\dots\lambda x_m.\left(\dots 
	\left(\left((\lambda z.N)P\right)Q_1\right)\dots \right)Q_n \]
	or:
	\[ M \quad\bisim\quad \lambda x_1\dots\lambda x_m.\left(\dots 
	\left((y)Q_1\right)\dots \right)Q_n \]
	where $m, n\in \NN$, $x_1,\dots,x_m,y,z\in\Vv$ and $N,P,Q_1,\dots,Q_n\in\linf{001}$.
	In the first case, $(\lambda y.N)P$ is the \emph{head redex} of $M$. In the second case, $M$ is in \emph{head normal form} (\textsc{hnf}).
	
	Similarly, a resource term $s\in\lr{}$ can always be written $s = \lambda x_1\dots\lambda x_m. \rapp{ \dots \rapp{ \rapp{u}\ms{t_1} } \dots } \ms{t_n}$ where $u$ is either a (head) redex or a variable.
\end{lem}

	\begin{proof}
	By induction, following the inductive structure of $M$ (we do not need to cross any coinductive rule here, and remain within the first \enquote{coinductive level} of $M$; thus, the proof is exactly the same as in the finitary case).
	\end{proof}

\begin{defi}[head reductions]
	The \emph{head reduction} is the relation $\flh$ defined on $\linf{001}$ so that
	$M\flh N$ if $N$ is obtained by reducing the head redex of $M$.

	Similarly, the \emph{resource head reduction} is the relation $\flrho$
	defined on $\lr{}$ such that $s\flrho T$ if
	$T$ is obtained by reducing the head redex of $s$.
	It is extended to $\flrh$ on $\lrsums{}$ in the same way as $\flr$.
\end{defi}

\begin{nota}[head reduction operators]
	If $M\in\linf{001}$ we set: $H(M) \defeq N$ if $M\flh N$; and $H(M) \defeq M$ if $M$ is in \textsc{hnf}.

	Similarly, if $s\in\lr{}$ we set: $H_r(s) \defeq T$ if $s\flrho T$; and $H(s) \defeq s$ if $s$ is in \textsc{hnf}.
	This operator is extended to $\lrsums{}$ by $H_r\left( \sum_i s_i \right) \defeq \sum_i H_r(s_i)$.
\end{nota}

\begin{lem}[simulation of the head reduction operator] \label{norm:lem:simul_flh}
	Let $M\in\linf{001}$ be a term, then
	\myorjourn{\\}{}
	$H_r(\Tay(M)) = \Tay(H(M))$.
\end{lem}
	
	\begin{proof}
	Direct consequence of \cref{simul:lem:simul_subst}.
	\end{proof}

\begin{lem}[termination of the resource head reduction operator] \label{norm:lem:termin_hr}
	Let $S \in \freemod{2}{\lr{}}$ be a sum of resource terms, then there exists $k \in \NN$ such that $H_r^k(S)$ is in \textsc{hnf}.
\end{lem}
	
	\begin{proof}
	Given $S \in \lrsums$, write $S = S' + S_{\text{\sc hnf}}$ where 
	$S_{\text{\sc hnf}}$ contains the terms of $S$ in \textsc{hnf}. By 
	definition of $H_r$, we have $H_r(S) = H_r(S') + S_{\text{\sc hnf}}$ with 
	$|H_r(S')| < |S'|$ from \cref{taylor:lem:size} whenever $S' \neq 0$,
	so that $\Vert H_r(S)\Vert \prec \Vert S\Vert$ in this case
	--- with the notation of \cref{taylor:lem:confl_term_flr}.
	We conclude by the well-foundedness of $\prec$.
	\end{proof}

Now we provide a characterisation of head-normalising infinitary terms based
on their Taylor expansion.
In a finitary setting, this result has been folklore for some time
\autocite{Olimpieri18,Olimpieri20}.

\begin{thm}[characterisation of head-normalising terms] \label{norm:thm:charac_head}
	Let $M\in\linf{001}$ be a term, then the following propositions are equivalent:
	\begin{enumerate}
	\item there exists $N\in\linf{001}$ in \textsc{hnf} such that $M\flbi N$,
	\item there exists $s\in\Tay(M)$ such that $\nf{r}{s} \neq 0$,
	\item there exists $N\in\linf{001}$ in \textsc{hnf} such that $M\flhs N$.
	\end{enumerate}
\end{thm}

	\begin{proof}
	Suppose (1), that is $M \flbi N \bisim \lambda x_1\dots\lambda 
	x_m.\left(\dots \left((y)N_1\right)\dots \right)N_n$. In particular, 
	$\Tay(N)$ contains $t_0 \bisim \lambda x_1\dots\lambda x_m. \rapp{ \dots 
	\rapp{ \rapp{y} 1 } \dots } 1$, which is normal.
	Using \cref{simul:thm:simul_flbi}, there 
	exists $s\in\Tay(M)$ and $T\in \Tay(N)$ such that $s\flrs t_0+T$ which 
	proves (2).
	
	Now, suppose that (2) holds, that is $s\flrs t_0+T$ with $t_0$ in normal 
	form. According to \cref{norm:lem:termin_hr}, there is a $k\in\NN$ such 
	that $H_r^k(s)$ is in \textsc{hnf}. Thus, using confluence, there exists a 
	$U\in \lrsums $ such that:
	$t_0+T\flrs U$ and $H_r^k(s)\flrs U$.
	Since $t_0$ is in normal form, $t_0\in U$. Thus, $H_r^k(s) \neq 0$, so there exists a term
		\[
		\lambda x_1\dots\lambda x_m. \rapp{ \dots \rapp{ \rapp{y}\ms{t_1} } \dots } \ms{t_n}
		\ \in\ H_r^k(s) 
		\ \in\ H_r^k(\Tay(M)) 
		\ =\ \Tay(H^k(M)),
		\]
	by \cref{norm:lem:simul_flh}. As a consequence, $H^k(M)$ has shape $\lambda x_1\dots\lambda x_m.\left(\dots \left((y)M_1\right)\dots \right)M_n$, which shows (3).
	
	Finally, (1) is as immediate consequence of (3).
	\end{proof}

A first notable consequence of the previous result is the equivalence of head-normalisation and solvability. In the finitary λ-calculus, this is a well-known theorem \autocite{Wadsworth76}. The following proof, based on the Taylor expansion and inspired by \autocite{Olimpieri20}, is much simpler than the original one.

\begin{defi}[solvability]
	A term $M\in\linf{001}$ is said to be \emph{solvable} in $\Lambda$
	(resp.  in $\linf{001}$) if there exist $x_1,\dots,x_m\in\Vv$ and
	$N_1,\dots,N_n\in\Lambda$ (resp. $\linf{001}$) such that
	\begin{flalign*}
		&& (\dots((\lambda x_1 \dots \lambda x_m. M) N_1)\dots) N_n
		&  \flbs \lambda x.x
		&  \text{(resp. $\flbi$).}
	\end{flalign*}
	Otherwise, $M$ is \emph{unsolvable}.
\end{defi}

\begin{cor}[characterisation of solvable terms]
	Let $M\in\linf{001}$ be a term, then the following propositions are equivalent:
	\begin{enumerate}
	\item $M$ is solvable in $\linf{001}$,
	\item $M$ is head-normalising,
	\item $M$ is solvable in $\Lambda$.
	\end{enumerate}
\end{cor}

	\begin{proof}
	Suppose (1), \emph{i.e.} there exists $x_1,\dots,x_m\in\Vv$ and $N_1,\dots,N_n\in\linf{001}$ such that 
	\[(\dots((\lambda x_1 \dots \lambda x_m. M) N_1)\dots) N_n \flbi \lambda x.x\,,\]
	which is in \textsc{hnf}.
	Then, according to \cref{norm:thm:charac_head}, there is an
	\[ s\in\Tay\left( (\dots((\lambda x_1 \dots \lambda x_m. M) N_1)\dots) N_n \right)\]
	such that $\nf{r}{s} \neq 0$.
	This resource term has shape 
	\[s \bisim \rapp{\dots\rapp{ \rapp{ \lambda x_1 \dots \lambda x_m.u } \ms{t_1} }\dots} \ms{t_n}\]
	with $u\in\Tay(M)$ and $\ms{t_i}\in\Tay(N_i)$.
	We must have $\nf{r}{u} \neq 0$, which leads to (2), 
	again with \cref{norm:thm:charac_head}.
	
	Now, suppose (2). \Cref{norm:thm:charac_head} gives
	$M\flhs \lambda x_1\dots\lambda x_m.\left(\dots \left((y)M_1\right)\dots \right)M_n$.
	Then:
	\begin{itemize}
	\item if $y \bisim x_i$, then $(M)(K^nI)^{(m)} \flbi I$,
	\item otherwise, $\left((\lambda y.M) K^nI \right) I^{(m)} \flbi I$,
	\end{itemize}
	using \cref{linf:nota:puissances} and the usual terms $I \defeq \lambda x.x$ and $K \defeq \lambda x.\lambda y.x$. This shows (3).
	
	The implication from (3) to (1) is direct, by \cref{linf:lem:flbi_transitive}.
	\end{proof}

\subsection{Normalisation, confluence and the Commutation Theorem}

We shall now address the key properties of normalisation and confluence in 
$\linf{001}$. It is known since Kennaway \emph{et al.}'s seminal paper 
\autocite{KennawayAl97} that, even though the infinitary λ-calculi are not 
strongly normalising (in any version $\linf{abc}$, there is no strongly 
convergent reduction from the term $\Omega$ to a normal form), the so-called 
β$\bot$-reduction is normalising and confluent in $\linf{001}$, $\linf{101}$ 
and $\linf{111}$. This is a confluence \enquote{up to a set of 
\emph{meaningless} terms}, which are forced to reduce to a constant $\bot$ 
(this technique was introduced by \autocite{Berarducci96} and 
\autocite{KennawayAl97,KennawayAl96}; for a summary, see 
\autocite[§~6.3]{BarendregtManzonetto}). In the case of $\linf{001}$, the 
meaningless terms are the unsolvable ones.

In this part, we use the Taylor expansion and a new version of Ehrhard and Regnier's Commutation theorem to give a simple presentation of normalisation, confluence, and a few other noteworthy corollaries.

First, we have to add the constant $\bot$ to our language, and to update the definition of the reductions and of the Taylor expansion correspondingly.

\begin{defi}[λ$\bot$-terms]
	Given a set of variables $\Vv$, the set $\lbinf{001}$ of 001-infinitary λ$\bot$-terms is defined by:
	\[ \lbinf{001} \defeq 
	\nu Y.\mu X.\left( \Vv + \lambda\Vv.X + (X)Y + \bot \right). \]
\end{defi}

\begin{defi}[β$\bot$-reduction]
	The binary relation $\bot_0$ is defined on $\lbinf{001}$ by:
	\[\bot_0 \defeq \{ (M,\bot),\ \text{$M$ is unsolvable}\}
	\cup \{(\lambda x.\bot, \bot),\ x\in\Vv\}
	\cup \{((\bot)M, \bot),\ M\in\lbinf{001}\}.\]
	
	The β$\bot$-reduction $\flbb$ is the contextual closure of $\beta_0\cup \bot_0$. The 
	infinitary β$\bot$-reduction $\flbbi$ is the 001-strongly convergent 
	closure of $\flbb$.
\end{defi}

Recall that, as underlined in \cref{linf:rem:nf},
a \emph{β$\bot$-normal form} is a term that cannot be reduced
through $\flbb$ (and not through $\flbbi$, which is reflexive),
and that the β$\bot$-normal forms of $M$ are the β$\bot$-normal terms $N$
such that $M \flbbi N$.

\begin{defi}[Taylor expansion for $\bot$]
	The Tayor expansion is extended to $\lbinf{001}$ by defining $\apptay$ exactly as in \cref{taylor:def:taylor}. This means that there is no approximant of $\bot$, and thereby $\Tay(\bot) = 0$.
\end{defi}

\begin{rem}\label{norm:rem:zero}
Observe that $\Tay(M)=0$ iff
$M\bisim\bot$ or $M\bisim\lambda x.M'$ or $M\bisim (M')N$ 
with $\Tay(M')=0$.
In particular, if $\Tay(M)=0$ and $M\not=\bot$ then $M$ contains a subterm of the form 
$\lambda x.\bot$ or $(\bot)N$.
\end{rem}

The following result is an extension of \cref{simul:thm:simul_flbi}, ensuring that adding the constant $\bot$ does not break all our previous work.

\begin{cor}[simulation of $\flbbi$] \label{norm:lem:simul_bbot}
	Let $M,N \in \lbinf{001}$ be λ$\bot$-terms. If $M\flbbi N$, then 
	$\Tay(M)\flrst \Tay(N)$.
\end{cor}
	
	\begin{proof}
	If $M$ is unsolvable then $M \mathrel{\bot_0} \bot$, and 
	there is no $N$ in \textsc{hnf} such that $M\flbs N$. From 
	\cref{norm:thm:charac_head}, it follows that $\forall s\in\Tay(M),\ 
	\nf{r}{s}=0$, that is to say $\Tay(M) \flrst \Tay(\bot)$. If $M\bisim 
	\lambda x.\bot$ or $M\bisim (\bot)M'$, then $\Tay(M) = \Tay(\bot) = 0$. 
	This extends \cref{simul:lem:simul_flbs} to the β$\bot$-reduction: if 
	$M\flbbs N$, then $\Tay(M)\flrst \Tay(N)$. The rest of the proof is 
	analogous to \cref{simul}. 
	\end{proof}

The next definition concerns Böhm trees. Based on an idea by Böhm
\autocite{Bohm68} and formally defined by Barendregt \autocite{Barendregt77},
Böhm trees were introduced as a notion of infinite normal form for the usual
λ-calculus, giving account of the (potentially) infinite behaviour of λ-terms.
They rely on a coinductive definition (probably the first one in the study
of the λ-calculus), and are the normal forms of $\linf{001}$.

\begin{defi}[Böhm tree]
	The \emph{Böhm tree} of a term $M\in\linf{001}$ is the λ$\bot$-term 
	$\BT{M}$ defined coinductively as follows:
	\begin{itemize}
	\item if $M$ is solvable and $M \flhs \lambda x_1\dots\lambda x_m.\left(\dots \left((y)M_1\right)\dots \right)M_n$, then:
		\[ \BT{M} \defeq \lambda x_1\dots\lambda x_m.\left(\dots \left((y)\BT{M_1}\right)\dots \right)\BT{M_n}, \]
	\item if $M$ is unsolvable, then $\BT{M} \defeq \bot$.
	\end{itemize}
	This definition is extended to $\lbinf{001}$
	by setting $\BT{\bot} \defeq \bot$.
\end{defi}

Notice again that every coinductive call to $\BT{-}$ occurs in the right side of an application, that is to say under a rule $\regle{\coI}$ carrying the $\later$ modality.

\begin{lem} \label{norm:lem:bt_nf}
	Let $M\in\lbinf{001}$ be a term
	\begin{enumerate}
	\item $\BT{M}$ is in β$\bot$-normal form.
	\item If $M$ is in β$\bot$-normal form then $\BT{M} \bisim M$.
	\end{enumerate}
\end{lem}
	
	\begin{proof}
		\begin{enumerate}
		\item By induction on the definition of $\flbb$,
		we show that for any λ$\bot$-terms $M$ and $N$, 
		if $M\flbb N$ then $M$ is not a Böhm tree:
		in the base case, it is sufficient to observe that 
		a Böhm tree is never a $\beta$-redex nor a $\bot$-redex;
		the contextuality cases are straightforward.
		\item By coinduction on the definition of $\BT{M}$,
		using the fact that a β$\bot$-normal term is either solvable
		or equal to $\bot$.
		\qedhere 
		\end{enumerate}
	\end{proof}

\begin{lem}[weak β$\bot$-normalisation] \label{norm:lem:bbot_wn}
	Let $M\in\lbinf{001}$ be a term, then $M\flbbi \BT{M}$.
	Furthermore, if $M \in \linf{001}$ and $\BT{M}\in\linf{001}$,
	then $M\flbi\BT{M}$.
\end{lem}
	
	\begin{proof}
	We build a derivation of $M\flbbi \BT{M}$ coinductively.
	If $M$ is either unsolvable or $\bot$,
	we have $M\flbbr \bot=\BT{M}$ by definition.
	Otherwise, we have 
	\[M \flhs \lambda x_1\dots\lambda x_m.\left(\dots \left((y)M_1\right)\dots \right)M_n\,.\]
	In this case, we apply $m$ times rule $\regle{\lambda_\beta^\infty}$,
	then $n$ times rule $\regle{@_\beta^\infty}$,
	and proceed coinductively to build derivations of 
	$M_i\flbi \BT{M_i}$ for $1\le i\le m$.

	If $\BT{M}\in\linf{001}$, then the first case of the construction never
	occurs, and we obtain $M\flbi\BT{M}$ instead.
	\end{proof}

The following two technical lemmas, already well-known
in a finitary setting \autocite[Facts 4.17 and 4.15]{Vaux19},
will be useful to show the unicity of β$\bot$-normal forms.

\begin{lem} \label{norm:lem:taylor_nf}
	Let $M\in\lbinf{001}$ be a term in β$\bot$-normal form, then $\Tay(M)$ is in normal form.
\end{lem}
	
	\begin{proof}
	By contraposition, if some $s\in\Tay(M)$ contains a redex then so does $M$.
	\end{proof}

\begin{lem}[injectivity, almost] \label{norm:lem:taylor_quasi_inj}
	Let $M,N\in\lbinf{001}$ be terms. If $\Tay(M) = \Tay(N)$, and if neither $M$ nor $N$
	contain a subterm of the form $\lambda x.\bot$ or $(\bot)M'$, then 
	$M\bisim N$.
\end{lem}
	
	\begin{proof}
	By nested induction and coinduction on the structure of $M$:
		\begin{itemize}
		\item If $M\bisim \bot$, then $\Tay(N)=\Tay(M)=0$, hence $N\bisim\bot$ 
			by \cref{norm:rem:zero}.
		\item Likewise, if $M\bisim x$, then $\Tay(N)=\Tay(M)=\{x\}$ so $x 
		\apptay N$, and thus $N\bisim x$.
		\item If $M\bisim \lambda x.M'$, then $\Tay(N)=\Tay(M)=\lambda x.\Tay(M')$.
			By assumption, $M'\nbisim \bot$ so, by \cref{norm:rem:zero},
			there exists $s \in \Tay(M')$ and then $\lambda x.s \apptay N$.
			Thus, $\exists N'\in\lbinf{001},\ N\bisim \lambda x.N'$.
			Furthermore, we must have $\Tay(M')=\Tay(N')$ so,
			by induction hypothesis, $M'\bisim N'$ and finally $M \bisim N$.
		\item If $M\bisim (M')M''$, then $\Tay(N)=\Tay(M)=\rapp{\Tay(M')}\Tay(M'')^!$.
			By assumption, $M' \nbisim \bot$ so, by \cref{norm:rem:zero},
			there exist $s \in \Tay(M')$ and $\ms t \in \Tay(M'')^!$,
			and then $\rapp{s}\ms{t} \apptay N$.
			Thus, there exist $N',N''\in\lbinf{001}$ such that $N\bisim (N')N''$.
			The fact that $\Tay(M)\not=0$ together with the injectivity of 
			$(s,\ms{t})\mapsto \rapp{s}{\ms{t}}$
			ensure that $\Tay(M')=\Tay(N')$ and $\Tay(M'')^! = \Tay(N'')^!$,
			hence $\Tay(M'') = \Tay(N'')$.
			We deduce $M'\bisim N'$ by induction
			and proceed coinductively to prove $M''\bisim N''$.
		\qedhere
		\end{itemize}
	\end{proof}

\begin{rem}
	The previous lemma does establish the injectivity of $\Tay(-)$ when
	restricted to $\linf{001}$.
\end{rem}

Now we have all the necessary material available,
we can state the Commutation theorem, as well as some corollaries. 
Contrary to its original formulation \autocite[see][p.~193]{EhrhardRegnier06},
no specific definition of $\Tay(\BT{M})$ is needed here
thanks to the extension of the Taylor expansion to infinitary λ$\bot$-terms.

\begin{thm}[Commutation theorem] \label{norm:thm:commut}
	For all term $M\in\lbinf{001}$, $\nft{r}{\Tay(M)} = \Tay(\BT{M})$.
\end{thm}
	
	\begin{proof}
	From \cref{norm:lem:bbot_wn}, we know that $M\flbbi \BT{M}$.
	Using the simulation theorem (\cref{norm:lem:simul_bbot}), 
	we deduce that $\Tay(M)\flrst \Tay(\BT{(M)})$, 
	which itself is in normal form because $\BT{M}$ is, 
	using \cref{norm:lem:taylor_nf,norm:lem:bt_nf}.
	This is the desired result (\emph{via} \cref{taylor:nota:nft}).
	\end{proof}

 From the Commutation theorem we can deduce the following two results, originally proved in \autocite{KennawayAl97} and later reformulated as a particular case of confluence modulo any set of \emph{strongly meaningless terms} \autocite{Czajka20,BarendregtManzonetto}. 

\begin{cor}[unicity of β$\bot$-normal forms] \label{norm:cor:bbot_unf}
	Let $M\in\lbinf{001}$ be a term,
	then $\BT{M}$ is its unique β$\bot$-normal form. 
	Furthermore, if $M \in \linf{001}$ and $\BT{M}\in\linf{001}$, 
	then the latter is the unique β-normal form of $M$.
\end{cor}
	
	\begin{proof}
	Suppose there is an $N\in\lbinf{001}$ in β$\bot$-normal form such that 
	$M\flbbi N$. Then by \cref{norm:lem:bt_nf,norm:thm:commut},
	$\Tay(N)
	= \Tay(\BT{N})
	= \nft{r}{\Tay(N)}
	= \nft{r}{\Tay(M)} 
	= \Tay(\BT{M})$.
	Since neither $\BT{M}$ nor $N$
	(that are in β$\bot$-normal form)
	can contain a subterm of the form $\lambda x.\bot$ or $(\bot)P$,
	we can apply \cref{norm:lem:taylor_quasi_inj} and obtain $N\bisim \BT{M}$.
	\end{proof}

\begin{rem}
	For terms $M,N \in \linf{001}$, write $M \eqtay N$ whenever
	$\nft{r}{\Tay(M)} = \nft{r}{\Tay(N)}$, and 
	$M \eqbohm N$ whenever $\BT{M}\bisim\BT{N}$.
	\cref{norm:thm:commut} entails that $M \eqbohm N$ implies $M \eqtay N$.
	The proof of the previous corollary can be adapted to obtain the reverse
	implication:
	if $M\eqtay N$ then $\Tay(\BT{M})=\Tay(\BT{N})$
	by \cref{norm:thm:commut}, and then \cref{norm:lem:taylor_quasi_inj}
	entails $\BT{M}=\BT{N}$.

	This means in particular that all the models mentioned in
	\cref{simul:rem:models} are \emph{sensible}, \emph{i.e.} they equate all
	unsolvable terms.
\end{rem}

\begin{cor}[confluence of the β$\bot$-reduction] \label{norm:cor:bbot_confl}
	The reduction $\flbbi$ is confluent.
\end{cor}
	
	\begin{proof}
	Given $M,N,N' \in \lbinf{001}$:
	\begin{center}
	\begin{tikzpicture}
	\path	( 0, 0)	node (m)	{ $M$ }
	++		( 2, 1)	node (n)	{ $N$ }
	++		( 0,-2)	node (p)	{ $N'$ }
	++		( 5, 0)	node (btp)	{ $\BT{N'}$ }
	++		( 0, 2)	node (btn)	{ $\BT{N}$ } ;
	\draw[->] (m) -- (n) 
		node [pos=0.9, below] { \scriptsize $\beta\bot$ }
		node [pos=0.8, above] { \scriptsize $\infty$ }		;
	\draw[->] (m) -- (p) 
		node [pos=0.8, below] { \scriptsize $\beta\bot$ }
		node [pos=0.9, above] { \scriptsize $\infty$ }		;
	\draw[->] (n) -- (btn) 
		node [pos=0.9, below] { \scriptsize $\beta\bot$ }
		node [pos=0.9, above] { \scriptsize $\infty$ }
		node [pos=0.45, below] { \small (Lem. \ref{norm:lem:bbot_wn}) } ;
	\draw[->] (p) -- (btp) 
		node [pos=0.9, below] { \scriptsize $\beta\bot$ }
		node [pos=0.9, above] { \scriptsize $\infty$ }
		node [pos=0.5, below] { \small (Lem. \ref{norm:lem:bbot_wn}) } ;
	\draw[double,double distance=3pt] (btn) -- (btp)
		node [pos=0.45, right] { \small (Cor. \ref{norm:cor:bbot_unf}) } ;
	\end{tikzpicture}\vspace*{-2\baselineskip}
	\end{center}
	\end{proof}

Another consequence is the following characterisation of normalising terms, which again is an infinitary counterpart to some folklore finitary result. Whereas the finitary case relies on \emph{positive} resource terms (terms with no occurrence of the empty multiset~$1$), we have to refine this concept by considering \emph{$d$-positive} terms, that is terms with no occurrence of $1$ at depth smaller than $d$.

\Needspace{2cm}
\begin{defi}[$d$-positive resource terms]
	Given an integer $d\in\NN$, the set $\lrpl{d}{}$ of \emph{$d$-positive} resource terms is defined inductively as follows:
	\begin{itemize}
	\item $\lrpl{0}{} \defeq \lr{}$,
	\item if $d\geq 1$, $\lrpl{d}{} \defeq \Vv \ |\ \lambda\Vv.\lrpl{d}{} \ |\ \rapp{\lrpl{d}{}}\lrpl{d-1}{!}$ with $\lrpl{d}{!} \defeq \Mfin(\lrpl{d}{}) \setminus \{1\}$.
	\end{itemize}
\end{defi}

\begin{cor}[characterisation of normalising terms] \label{norm:thm:charac_norm}
	Let $M\in\linf{001}$ be a term,
	then the following propositions are equivalent:
	\begin{enumerate}
	\item there exists $N\in\linf{001}$ in β-normal form such that $M\flbi N$,
	\item for any $d\in\NN$, there exists $s\in\Tay(M)$ such that $\nf{r}{s}$ contains a $d$-positive term.
	\end{enumerate}
\end{cor}
	
	\begin{proof}
	Suppose (1), that is to say $\BT{M}\in\linf{001}$ by 
	\cref{norm:cor:bbot_unf}. 
	In particular, $M$ is solvable.
	By induction on $d\in\NN$:
	\begin{itemize}
	\item Case $d=0$. By solvability of $M$ and 
	\cref{norm:thm:charac_head} there is an $s\in\Tay(M)$ such that $\nf{r}{s} 
	\neq 0$, \emph{i.e.} it contains a (0-positive) term.
	\item Case $d\geq 1$. Since $M$ is solvable, $M\flbs \lambda 
	x_1\dots\lambda x_m.\left(\dots \left((y)M_1\right)\dots \right)M_n$ and 
	$\BT{M} \bisim \lambda x_1\dots\lambda x_m.\left(\dots 
	\left((y)\BT{M_1}\right)\dots \right)\BT{M_n}$, with $M_i\flbi \BT{M_i}$ 
	and $\BT{M_i}\in\linf{001}$. By induction, for every $i$ there is an 
	$s_i\in\Tay(M_i)$ such that $\nf{r}{s_i}$ contains a $(d-1)$-positive 
	$t_i$. Then by \cref{simul:thm:simul_flbi} there are $s\in\Tay(M)$ and 
	$S,T\in \lrsums $ such that:
		\begin{align*}
		s \quad & \flrs\quad \lambda x_1\dots\lambda x_m. \rapp{ \dots \rapp{ \rapp{y}[s_1] } \dots } [s_n] + S \\
		 \quad & \flrs\quad \lambda x_1\dots\lambda x_m. \rapp{ \dots \rapp{ \rapp{y}[t_1] } \dots } [t_n] + T
		\end{align*}
	where $\lambda x_1\dots\lambda x_m. \rapp{ \dots \rapp{ \rapp{y}[t_1] } \dots } [t_n]$ is in normal form and $d$-positive.
	\end{itemize}
	
	Conversely, we suppose (2) and establish $\BT{M}\in\linf{001}$
	by nested induction and coinduction on $\BT{M}\in\lbinf{001}$.
	First note that (2) ensures that $\nft{r}{\Tay(M)}\not=0$.
	Then  \cref{norm:thm:charac_head} entails that $M$ is solvable, so
	\[M\flhs \lambda x_1\dots\lambda x_m.\left(\dots \left((y)M_1\right)\dots \right)M_n\,.\]
	Now fix $d\in\NN$:
	(2) ensures that we can find $s\in\Tay(M)$ and $t\in\nf{r}{s}$
	so that $t$ is $(d+1)$-positive.
	By \cref{simul:thm:simul_flbi}, $t\in\Tay(\BT{M})$ so it has the shape 
	\[t = \lambda x_1\dots\lambda x_m. \rapp{ \dots \rapp{ \rapp{y}\ms{t_1} } \dots } \ms{t_n}\]
	where $\ms{t_i}\in\lrpl{d}{!}$, 
	\emph{i.e.} each $\ms{t_i}$ contains a normal and $d$-positive $t_{i,1}$. 
	By \cref{simul:thm:simul_flbi} again, there are $s_i\in\Tay(M_i)$ such that 
	$t_{i,1} \in\nf{r}{s_i}$.
	We have thus proved that (2) is valid for each $M_i$.
	We proceed coinductively to establish $\BT{M_i}\in\linf{001}$.
	\end{proof}

If the finitary case, normalisation is also equivalent to the termination of the \emph{left-parallel} reduction strategy, which plays the same role as the head strategy in \cref{norm:thm:charac_head} \autocite[Thm.~4.10]{Olimpieri20}. In our setting, there is of course no finite reduction strategy reaching the normal form of a term. A characterisation of the 001-normalising terms, called \emph{hereditarily head-normalising} (\textsc{hhn}) in the literature, has been shown by Vial by means of infinitary non-idempotent intersection types \autocite{Vial17,Vial21}, thus answering to the so-called \enquote{Klop's problem}. However, there is no hope for an effective characterisation, since \textsc{hhn} terms are not recursively enumerable \autocite{Tatsuta08}.

\subsection{Infinitary contexts and the Genericity Lemma}

To conclude this paper, we use the previous results to extend to $\linf{001}$
a classical result in λ-calculus, the Genericity lemma
\autocite[Prop.~14.3.24]{Barendregt84}.
A similar extension has been proved using completely different techniques 
\autocites[§~5.3]{KennawayAl96}[Thm.~20]{Salibra00}.
The intuition behind this lemma is that
an unsolvable subterm of a normalising term
cannot contribute to its normal form (it is \emph{generic}).
This justifies that unsolvables are taken as a class of meaningless terms —
in fact, the unsolvables are the largest non-trivial set of (formally defined) 
meaningless terms \autocite{SeveriDeVries11,BarendregtManzonetto}.

\begin{defi}[context] \label{norm:def:context}
	The set $\context{\linf{001}}{\hole}$ of 001-infinitary \emph{contexts} is defined by:
	\[ \context{\linf{001}}{\hole} = \nu Y.\mu X.\left( \Vv + \lambda\Vv.X + (X)Y + \hole \right)\]
	where $\hole$ is a constant called the \enquote{hole} (contexts are \emph{not} quotiented by α-equivalence).
	
	Given a context $C\in\context{\linf{001}}{\hole}$ and a term $M\in\linf{001}$, we denote as $\context{C}{M}$ the term obtained by substituting $M$ for each occurrence of $\hole$ in $C$ --- like $C[M/\hole]$, but possibly capturing the free variables of $M$.
\end{defi}

\begin{defi}[resource context]
	The set $\rcontext{\lr{}}{\hole}$ of \emph{resource contexts} is defined, as in \cref{norm:def:context}, by adding the constant $\hole$ to $\lr{}$ (again, without quotienting by α-equivalence).
	
	Given a resource context $c\in\rcontext{\lr{}}{\hole}$ and a resource monomial $\ms{t}\in\lr{!}$, we denote as $\rcontext{c}{\ms{t}}$ the sum of resource terms obtained by substituting each occurrence of $\hole$ in $c$ with exactly one element of $\ms{t}$, or $0$ if the cardinality of $\ms{t}$ does not match the number of occurrences of $\hole$ --- again, like $c\rapp{\ms{t}/\hole}$, but possibly capturing the free variables of $\ms{t}$.
\end{defi}

The Taylor expansion is extended to $\Tay : \context{\linf{001}}{\hole} \to 
\Pp(\rcontext{\lr{}}{\hole})$ by setting $\Tay(\hole) \defeq \{\hole\}$.

\begin{lem} \label{norm:lem:taylor_context}
	Let $C\in\context{\linf{001}}{\hole}$ be a context and $M\in\linf{001}$ be a term. Then:
	\[ \Tay\left( \context{C}{M} \right) = \left\{ \rcontext{c}{\,\ms{t}\,},\ c\in\Tay(C),\ \ms{t}\in\Tay(M)^! \right\}. \]
\end{lem}
	
	\begin{proof}
	Direct consequence of \cref{simul:lem:simul_subst}.
	\end{proof}

\begin{lem}[characterisation of $\Tay$ by the $d$-positive elements] \label{norm:lem:charac_taylor_dpos}
	Let $M,N\in\linf{001}$ be terms. If for any $d\in\NN$ there exists a 
	$d$-positive $s_d\in\Tay(M)\cap\Tay(N)$, then $M\bisim N$.
\end{lem}
	
	\begin{proof}
	Under the hypothesis, we establish $M\bisim N$ by nested induction and
	coinduction on the structure of $M$.
	\begin{itemize}
	\item Case $M\bisim x$. For $d=0$, $\exists s_0 \in \Tay(M) \cap \Tay(N)$. 
	Since $s_0 \in \Tay(M)$, $s_0\bisim x$ so $N\bisim x$ too.
	\item Case $M\bisim \lambda x.M'$. Suppose $\forall d\in\NN,\ \exists s_d 
	\in \Tay(M) \cap \Tay(N)$. Since $s_d\in\Tay(M)$, $s_d\bisim \lambda 
	x.s'_d$ for some $d$-positive $s'_d$. $s_d\in\Tay(N)$, whence $N\bisim 
	\lambda x.N'$ for some $N'$, and $s'_d\in\Tay(N')$. By induction, $M'\bisim 
	N'$, so $M\bisim N$.
	\item Case $M\bisim (M')M''$. Suppose $\forall d\in\NN,\ \exists s_d \in 
	\Tay(M) \cap \Tay(N)$. Since $s_d\in\Tay(M)$, $s_d\bisim 
	\rapp{t_d}\ms{u}_d$ for some $t_d\in\lrpl{d}{}$ and 
	$\ms{u}_d\in\lrpl{(d-1)}{!}$. Furthermore $s_d\in\Tay(N)$, whence $N\bisim 
	(N')N''$ for some $N'$ and $N''$ such that $t_d\in\Tay(N')$ and 
	$\ms{u}_d\in\Tay(N'')^!$. Since $\forall d\in\NN,\ t_d 
	\in \Tay(M')\cap\Tay(N')$ the induction hypothesis gives $M'\bisim N'$.
	Moreover, $\forall d\in\NN$, by $(d+1)$-positivity of $s_{d+1}$,
	$\ms{u}_{d+1}$ must contain at least one element $u_{d+1,1}$, which is
	$d$-positive and such that $u_{d+1,1}\in \Tay(M'')\cap\Tay(N'')$.
	Thus we can proceed to establish $M''\bisim N''$ coinductively.
	\qedhere
	\end{itemize}
	\end{proof}

Using the previous work, we can state and show the infinitary Genericity lemma 
— without any further hypotheses than in the finitary setting. Our proof is a 
refinement of the (finitary) proof by Barbarossa and Manzonetto 
\autocite[Thm.~5.3]{BarbarossaManzonetto20}. As stressed by the authors, the 
key feature of the Taylor expansion here is that a resource term cannot erase 
any of its subterms (without being itself reduced to zero). However, in the 
infinitary setting, a term is in general not characterised by a single element 
of its Taylor expansion, which motivates the above characterisation by 
$d$-positive elements.

\begin{thm}[Genericity lemma] \label{norm:thm:genericity}
	Let $M\in\linf{001}$ be an unsolvable term
	and $\context{C}{\hole}$ be a context in $\linf{001}$. 
	If $\context{C}{M}$ has a β-normal form $C^*$, 
	then for any term $N\in\linf{001}$, $\context{C}{N} \flbi C^*$.
\end{thm}
	
	\begin{proof}
	Suppose $\context{C}{M} \flbi C^*$ in β-normal form. Then:
	\begin{flalign*}
	&& \forall d\in\NN,\ \exists s\in\Tay(\context{C}{M}),\
		& \exists t_d\in\lrpl{d}{},\ t_d\in\nf{r}{s}
		& \text{by \cref{norm:thm:charac_norm}} \\
	\text{hence}\\
	&& \forall d\in\NN,\ \exists c\in\Tay(C),\ \exists\ms{m}\in\Tay(M)^!,\
		& \exists t_d\in\lrpl{d}{},\ t_d\in\nf{r}{\rcontext{c}{\ms{m}}}
		& \text{by \cref{norm:lem:taylor_context}.}
	\end{flalign*}
	
	Write $\ms m = [m_1,\ldots,m_n]$ with $n \defeq \deg_{\hole}(c)$.
	By unsolvability of $M$ and \cref{norm:thm:charac_head},
	for each $1\le i\le n$, $m_i \flrs 0$, so by confluence (\cref{taylor:lem:confl_term_flr}) there is a $T_d\in \lrsums $ for each $d\in\NN$ such that:
	\begin{center}
	\begin{tikzpicture}
	\path (0, 0)	node (cm)	{ $\rcontext{c}{\ms{m}}$ }
	+ (2.5,-1.5)	node (c0)	{
		$\rcontext{c}{[\underbrace{0,\ldots,0}_{\text{$n$ times}}]}$
		}
	+ (0,-3)	node (t)	{ $t_d+T_d$ } ;
	\draw[->] (cm) -- (c0) 
		node [pos=0.95, left=2pt] { \scriptsize $r$ }
		node [pos=0.95, above] { \scriptsize $*$ } ;
	\draw[->] (cm) -- (t) 
		node [pos=0.9, left] { \scriptsize $r$ }
		node [pos=0.9, right] { \scriptsize $*$ } ;
	\draw[->] (c0) -- (t) 
		node [pos=0.9, right=2pt] { \scriptsize $r$ }
		node [pos=0.9, above] { \scriptsize $*$ } ;
	\end{tikzpicture}
	\end{center}
	
	If $\hole$ appeared in $c$, then $n = \deg_{\hole}(c) \geq 1$ and 
	$\rcontext{c}{[0,\ldots,0]} = 0$, which is impossible. Thus, there is no 
	occurrence of $\hole$ in $c$, and $\rcontext{c}{1}$ is (the α-equivalence class of) $c$.
	
	Now, take any $N \in \linf{001}$.
	By \cref{norm:lem:taylor_context}, $\rcontext{c}{1} \apptay \context{C}{N}$.
	Since $\rcontext{c}{1} \flrs t_d + T_d$ we have
	$t_d \in \nft{r}{\Tay(\context{C}{N})} = \Tay\left(\BT{\context{C}{N}}\right)$
	by \cref{norm:thm:commut}.
	Similarly, since we had taken $t_d \in \nf{r}{s}$ for some $s \apptay \context{C}{M}$,
	$t_d \in \Tay\left(\BT{\context{C}{M}}\right) = \Tay(C^*)$,
	recalling that $\BT{\context{C}{M}}=C^*$ by \cref{norm:cor:bbot_unf}.
	
	There is a $d$-positive $t_d\in\Tay(\BT{\context{C}{N}})$ for any 
	$d\in\NN$, so we can apply \cref{norm:thm:charac_norm} and deduce that 
	$\BT{\context{C}{N}}\in\linf{001}$. Since in addition we have 
	$t_d\in\Tay(C^*)$, we obtain $\BT{\context{C}{N}}\bisim C^*$ by 
	\cref{norm:lem:charac_taylor_dpos} and $\context{C}{N} \flbi C^*$ by 
	\cref{norm:lem:bbot_wn}.
	\end{proof}

\section{Conclusion}

\paragraph{Summary.}

As the main result of this paper, we showed that the resource reduction of Taylor expansions simulates the infinitary β-reduction of $\linf{001}$ terms (\cref{simul:thm:simul_flbi}). This could be expected from Ehrhard and Regnier's Commutation Theorem, which tightly relates normalisation of the Taylor expansion and normal forms of $\lbinf{001}$ (\emph{aka.} Böhm trees), but remains remarkable in that it enables to simulate an infinitary dynamics with a finitary one.

Using this fact, we were able to give simple proofs of well-known properties of 
$\lbinf{001}$ like confluence (\cref{norm:cor:bbot_confl}), weak normalisation 
(\cref{norm:lem:bbot_wn}), unicity of normal forms (\cref{norm:cor:bbot_unf}). 
We also extended to infinitary terms several λ-calculus results like the 
Commutation Theorem (\cref{norm:thm:commut}), the characterisations of head- 
and β-normalisation through Taylor expansion 
(\cref{norm:thm:charac_head,norm:thm:charac_norm}),
and we provided a new proof of 
the infinitary Genericity Lemma (\cref{norm:thm:genericity}).

As we already underlined, we believe that these results suggest that $\lbinf{001}$ is a reasonable extension of $\Lambda$ to consider when adressing head-normalisation and Taylor expansion issues. In particular, we were able to express the Commutation Theorem without any technical patch for the treatment of Böhm trees and reduction towards them.

\paragraph{Further work.}

The question naturally arises whether the converse of 
\cref{simul:thm:simul_flbi} is also true, 
that is whether $M \flbi N$ whenever $\Tay(M) \flrst \Tay(N)$. 
Similar issues have been successfully addressed in the setting of the algebraic
λ-calculus \autocite{Kerinec19,KerinecVaux23}, 
which suggests such a conservativity result is within reach.

It is in fact possible to show that for ordinary λ-terms
$M,N\in\Lambda$, $\Tay(M) \flrst \Tay(N)$ implies $M\flbs N$.
In the infinitary setting, however, the conjecture fails:
we were able to design terms $A, \bar{A} \in \linf{001}$
such that $\Tay(A) \flrst \Tay(\bar{A})$,
and such that there exists no reduction $A \flbi \bar{A}$.
These results are the subject of a separate paper 
\autocite{CerdaVaux23_acc}.

\medskip

One could also ask whether the Taylor expansion can be further extended
to the $\linf{101}$ and $\linf{111}$ infinitary calculi,
looking for a counterpart to the Commutation Theorem
involving Lévy-Longo and Berarducci trees.
We believe this can be done with only minor adaptions
in the case of Lévy-Longo trees,
but not in the case of Berarducci trees.

Indeed, recall that Böhm trees can be seen as maximal directed sets
of finite λ$\bot$-terms in β$\bot$-normal form,
\ie in \textsc{hnf} \autocite[Sec.~14.3]{Barendregt84}.
The crucial observation by \autocite{EhrhardRegnier08}
is that such terms are isomorphic to affine resource terms in normal form,
the isomorphism mapping the elements of $\BT{M}$
to the affine elements of $\Tay(\BT{M})$.
\begin{itemize}
\item It is easy to design a resource calculus extending this property
	to Lévy-Longo trees (and the corresponding β$\bot$-normal forms, 
	namely \emph{weak head} normal forms):
	one should add a constructor approximating an abstraction with unknown body
	as follows
	\begin{align*}
		\lr{} & \defeq \Vv \ |\ \lambda\Vv.\bullet \ |\ \lambda\Vv.\lr{} \ |\ 
		\rapp{\lr{}}\lr{!} \\
		\lr{!} &\defeq \Mfin(\lr{})\quad
	\end{align*}
	so that we could set $\Tay(\lambda x.M) \defeq \lambda x.\Tay(M) + \lambda x.\bullet$ 
	in order to take into account the possibility to encounter 
	an infinite chain of abstractions.
\item On the other hand, the notion of β$\bot$-normal form
	corresponding to Berarducci trees (\emph{top} normal forms)
	does not immediately enjoy such a property
	because there is no \enquote{top-level} syntactic characterisation
	of top normal forms:
	$(M)N$ is in \textsc{tnf} if $M$ does not reduce to an abstraction,
	which can only be checked by reducing $M$ at an unknown depth.
\end{itemize}
Thus, designing a Taylor approximation for the $\linf{111}$ calculus,
if possible, seems to require more advanced techniques.

\medskip

Finally, we have limited our study to a qualitative setting only:
as explained in \cref{taylor:rem:coef}, it is not difficult 
to extend the definition of Taylor expansion with appropriate coefficients;
but as explained in \cref{taylor:rem:quantitative}, a quantitative version
of our simulation result seems out of reach, if only because the reduction of
infinite weighted sums of resource terms is not well defined in general.
Nonetheless, we conjecture that the Commutation theorem also holds in a
quantitative setting.

Indeed, in their seminal results \autocite{EhrhardRegnier06,EhrhardRegnier08},
Ehrhard and Regnier exploited a uniformity property to show that
the normalization (rather than an arbitrary reduction) of a sum 
of resource terms obtained by Taylor expansion does not generate sums of
coefficients: each term occurring in the normal form is generated by a single
term of the original sum.
It is then possible to deduce the quantitative Commutation theorem from the
qualitative one: this was essentially the path followed by Ehrhard and Regnier,
and revisited by Olimpieri and the second author \autocite{OlimpieriVaux22}.
In the latter work, the qualitative Commutation theorem was established
quite straightforwardly, by proving that Taylor expansion commutes 
with a variant of hereditary head reduction (the reduction strategy 
underlying the definition of Böhm trees).
Moreover, by contrast with arbitrary reduction of resource terms,
the latter reduction strategy does enjoy the uniformity property.
We do believe that this alternative approach can be adapted to the infinitary
setting, in order to deal with quantitative Taylor expansion:
we leave this for future work.

\printbibliography

\end{document}